\documentclass[preprint,12pt]{elsarticle}

\usepackage{fullpage,amssymb,amsthm,caption,subcaption,enumerate,graphicx}
\usepackage[nosumlimits]{amsmath}
\usepackage[nothing]{algorithm}
\usepackage{algorithmicx}
\usepackage[noend]{algpseudocode}
\usepackage{array}

\floatname{algorithm}{Pseudocode}

\allowdisplaybreaks

\journal{Journal of Discrete Algorithms}

\renewcommand{\paragraph}[1]{{\medskip\noindent{\bf #1}}}

\newcommand{\emparagraph}[1]{{\medskip\noindent{\it #1}}}
\newcommand{\etal}{{\it et al.}}

\newcommand{\mycase}[1]{\mbox{{\underline{Case #1}}:\/}}

\newcommand{\margincomment}[1]%
    {{%
      \marginpar{{\tiny\begin{minipage}{0.5in}
                       \begin{flushleft}
                          {#1}
                       \end{flushleft}
                       \end{minipage}
                }}
    }}

\newcommand{\hatr}{{\hat r}}
\newcommand{\hatx}{{\hat x}}
\newcommand{\haty}{{\hat y}}
\newcommand{\dotx}{{\dot x}}
\newcommand{\doty}{{\dot y}}
\newcommand{\dotr}{{\dot r}}

\newcommand{\barx}{{\bar x}}
\newcommand{\bary}{{\bar y}}

\newcommand{\bard}{{\bar d}}

\newcommand{\wbarN}{{\overline{N}}}

\newcommand{\bfr}{\boldsymbol{r}}

\newcommand{\bfx}{\boldsymbol{x}}
\newcommand{\bfy}{\boldsymbol{y}}

\newcommand{\bfalpha}{\boldsymbol{\alpha}}
\newcommand{\bfbeta}{\boldsymbol{\beta}}

\newcommand{\calA}{{\cal A}}

\newcommand{\hatcalI}{{\hat{\cal I}}}

\newcommand{\dotcalI}{{\dot{\cal I}}}

\newcommand{\tilder}{{\widetilde r}}
\newcommand{\tildex}{{\widetilde x}}
\newcommand{\wtildeN}{{\widetilde N}}
\newcommand{\tildebfr}{\widetilde{\boldsymbol{r}}}
\newcommand{\tildebfx}{\widetilde{\boldsymbol{x}}}

\newcommand{\barbfx}{\bar{\boldsymbol{x}}}
\newcommand{\barbfy}{\bar{\boldsymbol{y}}}
\newcommand{\hatbfx}{\hat{\boldsymbol{x}}}
\newcommand{\hatbfy}{\hat{\boldsymbol{y}}}
\newcommand{\dotbfx}{\dot{\boldsymbol{x}}}
\newcommand{\dotbfy}{\dot{\boldsymbol{y}}}

\newcommand{\half}{{\mbox{$\frac{1}{2}$}}}

\newcommand{\onethird}{{\mbox{$\frac{1}{3}$}}}

\newcommand{\fourthirds}{{\mbox{$\frac{4}{3}$}}}

\newcommand{\braced}[1]{{ \left\{ #1 \right\} }}

\newcommand{\floor}[1]{{ \lfloor #1 \rfloor }}

\newcommand{\set}{\,{\leftarrow}\,}
\newcommand{\suchthat}{{\,:\,}}
\newcommand{\cost}{{\it cost}}

\newcommand{\assign}{{\,\leftarrow\,}}

\newcommand{\NP}{{\mbox{\sf NP}}}
\newcommand{\PP}{{\mbox{\sf P}}}
\newcommand{\DTIME}{{\mbox{\sf DTIME}}}

\DeclareMathOperator*{\argmin}{arg\,min}

\newcommand{\LP}{\mbox{\rm LP}}
\newcommand{\OPT}{\mbox{\rm OPT}}

\newcommand{\EGUP}{\mbox{\rm EGUP}} 
\newcommand{\ECHS}{\mbox{\rm ECHS}} 
\newcommand{\EBGS}{\mbox{\rm EBGS}} 

\newcommand{\smallEGUP}{\mbox{\tiny\rm EGUP}}
\newcommand{\smallECHS}{\mbox{\tiny\rm ECHS}}
\newcommand{\smallEBGS}{\mbox{\tiny\rm EBGS}}

\newcommand{\FTFP}{\mbox{\rm FTFP}}

\newcommand{\calI}{\mathcal{I}}
\newcommand{\avg}{{\mbox{\scriptsize\rm avg}}}
\newcommand{\cls}{{\mbox{\scriptsize\rm cls}}}
\newcommand{\far}{{\mbox{\scriptsize\rm far}}}

\newcommand{\smallfar}{{\mbox{\tiny\rm far}}}

\newcommand{\sitesset}{\mathbb{F}}
\newcommand{\clientset}{\mathbb{C}}
\newcommand{\facilityset}{\overline{\sitesset}}
\newcommand{\demandset}{\overline{\clientset}}

\newcommand{\concost}{C^{\avg}}

\newcommand{\tcc}{\mbox{\rm{tcc}}}
\newcommand{\clsdist}{C_{\cls}^{\avg}}
\newcommand{\fardist}{C_{\far}^{\avg}}
\newcommand{\clsmax}{C_{\cls}^{\max}}

\newcommand{\wbarclsnb}{\wbarN_{\cls}}
\newcommand{\wbarfarnb}{\wbarN_{\far}}
\newcommand{\wtildeclsnb}{\wtildeN_{\cls}}
\newcommand{\tcccls}{\mbox{\rm{tcc}}_{\cls}}
\newcommand{\dmaxcls}{\mbox{\rm{dmax}}_{\cls}}

\newcommand{\Exp}{\mathbb{E}}
\newcommand{\Prob}{\mathbb{P}}

\newcommand{\NearestUnitChunk}{{\textsc{NearestUnitChunk}}}
\newcommand{\AugmentToUnit}{{\textsc{AugmentToUnit}}}
\newcommand{\connsum}{{\textrm{conn}}}

\newtheorem{theorem}{Theorem}

\newtheorem{corollary}[theorem]{Corollary}
\newtheorem{lemma}[theorem]{Lemma}

\newcommand{\ignore}[1]{}

\begin{document}

\begin{frontmatter}
  \title{LP-rounding Algorithms for the Fault-Tolerant\\
    Facility Placement Problem\tnoteref{t1}} \tnotetext[t1]{A
    preliminary version of this article appeared in Proc. CIAC 2013.}

\author{Li Yan\corref{cor1} and Marek Chrobak\fnref{fn1}}
\ead{\{lyan,marek\}@cs.ucr.edu}
\address{Department of Computer Science\\
 University of California at Riverside\\
Riverside, CA 92521, USA}

\cortext[cor1]{Corresponding author}
\fntext[fn1]{Work supported by NSF grants CCF-0729071 and CCF-1217314.}
\begin{abstract} 
  The Fault-Tolerant Facility Placement problem (FTFP) is a
  generalization of the classic Uncapacitated Facility
  Location Problem (UFL). In FTFP we are given a set of
  facility sites and a set of clients. Opening a facility at
  site $i$ costs $f_i$ and connecting client $j$ to a
  facility at site $i$ costs $d_{ij}$. We assume that the
  connection costs (distances) $d_{ij}$ satisfy the triangle
  inequality. Multiple facilities can be opened at any
  site. Each client $j$ has a demand $r_j$, which means that
  it needs to be connected to $r_j$ different facilities
  (some of which could be located on the same site). The
  goal is to minimize the sum of facility opening cost and
  connection cost.

  The main result of this paper is a $1.575$-approximation algorithm
  for FTFP, based on LP-rounding. The algorithm first reduces the
  demands to values polynomial in the number of sites. Then it uses a
  technique that we call adaptive partitioning, which partitions the
  instance by splitting clients into unit demands and creating a
  number of (not yet opened) facilities at each site. It also
  partitions the optimal fractional solution to produce a fractional
  solution for this new instance.  The partitioned fractional solution
  satisfies a number of properties that allow us to exploit existing
  LP-rounding methods for UFL to round our partitioned solution to an
  integral solution, preserving the approximation ratio.  In
  particular, our $1.575$-approximation algorithm is based on the
  ideas from the $1.575$-approximation algorithm for UFL by
  Byrka~\etal, with changes necessary to satisfy the fault-tolerance
  requirement.
\end{abstract}

\begin{keyword}
Facility Location \sep Approximation Algorithms
\end{keyword}

\end{frontmatter}


\section{Introduction}

In the \emph{Fault-Tolerant Facility Placement} problem
(FTFP), we are given a set $\sitesset$ of \emph{sites} at
which facilities can be built, and a set $\clientset$ of
\emph{clients} with some demands that need to be satisfied
by different facilities. A client $j\in\clientset$ has demand
$r_j$. Building one facility at a site $i\in\sitesset$ incurs a cost
$f_i$, and connecting one unit of demand from client $j$ to
a facility at site $i$ costs $d_{ij}$. Throughout the
paper we assume that the connection costs (distances)
$d_{ij}$ form a metric, that is, they are
symmetric and satisfy the triangle inequality. In a feasible solution, some
number of facilities, possibly zero, are opened at each site
$i$, and demands from each client are connected to those
open facilities, with the constraint that demands from the
same client have to be connected to different
facilities. Note that any two facilities at the same site are considered different.

It is easy to see that if all $r_j=1$ then FTFP reduces to
the classic Uncapacitated Facility Location problem (UFL).
If we add a constraint that each site can have at most one
facility built on it, then the problem becomes equivalent to the
Fault-Tolerant Facility Location problem (FTFL). One
implication of the one-facility-per-site restriction in FTFL
is that $\max_{j\in\clientset}r_j \leq |\sitesset|$, while
in FTFP the values of $r_j$'s can be much bigger than
$|\sitesset|$.

The UFL problem has a long history; in particular, great
progress has been achieved in the past two decades in
developing techniques for designing constant-ratio
approximation algorithms for UFL.  Shmoys, Tardos and
Aardal~\cite{ShmoysTA97} proposed an approach based on
LP-rounding, that they used to achieve a ratio of 3.16.
This was then improved by Chudak~\cite{ChudakS04} to 1.736,
and later by Sviridenko~\cite{Svi02} to 1.582.
The best known ``pure" LP-rounding algorithm is due to
Byrka~{\etal}~\cite{ByrkaGS10} with ratio 1.575. 
Byrka and Aardal~\cite{ByrkaA10} gave a hybrid algorithm that combines LP-rounding
and dual-fitting (based on \cite{JainMMSV03}), achieving a ratio of 1.5.  Recently,
Li~\cite{Li11} showed that, with a more refined analysis and
randomizing the scaling parameter used in \cite{ByrkaA10}, the ratio can be improved
to 1.488. This is the best known approximation result for UFL.  
Other techniques include the primal-dual algorithm with ratio 3 by
Jain and Vazirani~\cite{JainV01}, the dual fitting method by
Jain~{\etal}~\cite{JainMMSV03} that gives ratio 1.61, and a
local search heuristic by Arya~{\etal}~\cite{AryaGKMMP04}
with approximation ratio 3.  On the hardness side, UFL is
easily shown to be {\NP}-hard, and it is known that it is
not possible to approximate UFL in polynomial time with
ratio less than $1.463$, provided that
$\NP\not\subseteq\DTIME(n^{O(\log\log
  n)})$~\cite{GuhaK98}. An observation by Sviridenko
strengthened the underlying assumption to $\PP\ne \NP$ (see \cite{vygen05}).

FTFL was first introduced by Jain and
Vazirani~\cite{JainV03} and they adapted their primal-dual
algorithm for UFL to obtain a ratio of
$3\ln(\max_{j\in\clientset}r_j)$.  All subsequently
discovered constant-ratio approximation algorithms use
variations of LP-rounding.  The first such algorithm, by
Guha~{\etal}~\cite{GuhaMM01}, adapted the approach for UFL
from \cite{ShmoysTA97}.  Swamy and Shmoys~\cite{SwamyS08}
improved the ratio to $2.076$ using the idea of pipage
rounding introduced in \cite{Svi02}. Most recently,
Byrka~{\etal}~\cite{ByrkaSS10} improved the ratio to 1.7245
using dependent rounding and laminar clustering.

FTFP is a natural generalization of UFL. It was first
studied by Xu and Shen~\cite{XuS09}, who extended the
dual-fitting algorithm from~\cite{JainMMSV03} to give an
approximation algorithm with a ratio claimed to be
$1.861$. However their algorithm runs in polynomial time
only if $\max_{j\in\clientset} r_j$ is polynomial in
$O(|\sitesset|\cdot |\clientset|)$ and the analysis of the
performance guarantee in \cite{XuS09} is flawed\footnote{Confirmed through
  private communication with the authors.}.  To date, the
best approximation ratio for FTFP in the literature is $3.16$,
established by Yan and Chrobak~\cite{YanC11}, while the only
known lower bound is the $1.463$ lower bound for UFL
from~\cite{GuhaK98}, as UFL is a special case of FTFP.
If all demand values $r_j$ are equal, the problem can be solved
by simple scaling and applying LP-rounding algorithms for UFL. This does
not affect the approximation ratio, thus achieving ratio $1.575$ for this
special case (see also \cite{LiaoShen11}).

\smallskip

The main result of this paper is an LP-rounding algorithm
for FTFP with approximation ratio 1.575, matching the best
ratio for UFL achieved via the LP-rounding method
\cite{ByrkaGS10} and significantly improving our earlier
bound in~\cite{YanC11}. In Section~\ref{sec: polynomial
  demands} we prove that, for the purpose of LP-based
approximations, the general FTFP problem can be reduced to
the restricted version where all demand values are
polynomial in the number of sites.  This \emph{demand
  reduction} trick itself gives us a ratio of $1.7245$,
since we can then treat an instance of FTFP as an instance
of FTFL by creating a sufficient (but polynomial) number of
facilities at each site, and then using the algorithm
from~\cite{ByrkaSS10} to solve the FTFL instance.

The reduction to polynomial demands suggests an approach where
clients' demands are split into unit demands. These unit demands can
be thought of as ``unit-demand clients'', and a natural approach would
be to adapt LP-rounding methods from
\cite{gupta08,ChudakS04,ByrkaGS10} to this new set of unit-demand
clients.  Roughly, these algorithms iteratively pick a client that
minimizes a certain cost function (that varies for different
algorithms) and open one facility in the neighborhood of this
client. The remaining clients are then connected to these open
facilities.  In order for this to work, we also need to convert the
optimal fractional solution $(\bfx^\ast,\bfy^\ast)$ of the original
instance into a solution $(\barbfx,\barbfy)$ of the modified instance
which then can be used in the LP-rounding process. This can be thought
of as partitioning the fractional solution, as each connection value
$x^\ast_{ij}$ must be divided between the $r_j$ unit demands
of client $j$ in some way. In Section~\ref{sec: adaptive partitioning} we
formulate a set of properties required for this partitioning to
work. For example, one property guarantees that we can connect demands
to facilities so that two demands from the same client are connected
to different facilities. Then we present our \emph{adaptive
  partitioning} technique that computes a partitioning with all the
desired properties. Using adaptive partitioning we were able to extend
the algorithms for UFL from \cite{gupta08,ChudakS04,ByrkaGS10} to
FTFP. We illustrate the fundamental ideas of our approach in
Section~\ref{sec: 3-approximation}, showing how they can be used to
design an LP-rounding algorithm with ratio $3$.  In Section~\ref{sec:
  1.736-approximation} we refine the algorithm to improve the
approximation ratio to $1+2/e\approx 1.736$.  Finally, in
Section~\ref{sec: 1.575-approximation}, we improve it even further to
$1.575$ -- the main result of this paper.

Summarizing, our contributions are two-fold: One, we show
that the existing LP-rounding algorithms for UFL can be
extended to a much more general problem FTFP, retaining the
approximation ratio. We believe that, should even better
LP-rounding algorithms be developed for UFL in the future,
using our demand reduction and adaptive partitioning
methods, it should be possible to extend them to FTFP.
In fact, some improvement of the ratio
should be achieved by randomizing the scaling parameter
$\gamma$ used in our algorithm, as Li showed in \cite{Li11}
for UFL.  (Since the ratio $1.488$ for UFL in~\cite{Li11}
uses also dual-fitting
algorithms~\cite{MahdianYZ06}, we would not obtain the same
ratio for FTFP yet using only LP-rounding.)

Two, our ratio of $1.575$ is significantly better than the
best currently known ratio of $1.7245$ for the
closely-related FTFL problem. This suggests that in the
fault-tolerant scenario the capability of creating
additional copies of facilities on the existing sites makes
the problem easier from the point of view of approximation.

\section{The LP Formulation}\label{sec: the lp formulation}

The FTFP problem has a natural Integer Programming (IP)
formulation. Let $y_i$ represent the number of facilities
built at site $i$ and let $x_{ij}$ represent the number of
connections from client $j$ to facilities at site $i$. If we
relax the integrality constraints, we obtain the following LP:

\begin{alignat}{3}
  \textrm{minimize} \quad \cost(\bfx,\bfy) &= \textstyle{\sum_{i\in \sitesset}f_iy_i 
								+ \sum_{i\in \sitesset, j\in \clientset}d_{ij}x_{ij}}\label{eqn:fac_primal}\hspace{-1.5in}&&
									\\ \notag
  \textrm{subject to}\quad y_i - x_{ij} &\geq 0 			&\quad 		&\forall i\in \sitesset, j\in \clientset 
									\\ \notag
     \textstyle{\sum_{i\in \sitesset} x_{ij}} &\geq r_j  &			&\forall j\in \clientset
 									\\ \notag
  	  x_{ij} \geq 0, y_i &\geq 0 						& 			&\forall i\in \sitesset, j\in \clientset 
  									\\ \notag
\end{alignat}


\noindent
The dual program is:

\begin{alignat}{3}
  \textrm{maximize}\quad \textstyle{\sum_{j\in \clientset}} r_j\alpha_j&\label{eqn:fac_dual}  
     						\\ \notag
  \textrm{subject to} \quad \textstyle{
    \sum_{j\in \clientset}\beta_{ij}} &\leq f_i  &\quad\quad			&\forall i \in \sitesset  
							\\ \notag
  \alpha_{j} - \beta_{ij} 	&\leq  d_{ij}       &                 & \forall i\in \sitesset, j\in \clientset 
							\\ \notag
  \alpha_j \geq 0, \beta_{ij} &\geq 0           &            & \forall i\in \sitesset, j\in \clientset
  							\\ \notag
\end{alignat}

In each of our algorithms we will fix some optimal
solutions of the LPs (\ref{eqn:fac_primal}) and (\ref{eqn:fac_dual})
that we will denote by $(\bfx^\ast, \bfy^\ast)$ and
$(\bfalpha^\ast,\bfbeta^\ast)$, respectively.

With $(\bfx^\ast, \bfy^\ast)$ fixed, we can define the
optimal facility cost as $F^\ast=\sum_{i\in\sitesset} f_i
y_i^\ast$ and the optimal connection cost as $C^\ast =
\sum_{i\in\sitesset,j\in\clientset} d_{ij}x_{ij}^\ast$.
Then $\LP^\ast = \cost(\bfx^\ast,\bfy^\ast) = F^\ast+C^\ast$
is the joint optimal value of (\ref{eqn:fac_primal}) and
(\ref{eqn:fac_dual}).  We can also associate with each
client $j$ its fractional connection cost $C^\ast_j =
\sum_{i\in\sitesset} d_{ij}x_{ij}^\ast$.  Clearly, $C^\ast =
\sum_{j\in\clientset} C^\ast_j$.  Throughout the paper we
will use notation $\OPT$ for the optimal integral solution
of (\ref{eqn:fac_primal}).  $\OPT$ is the value we wish to
approximate, but, since $\OPT\ge\LP^\ast$, we can instead use
$\LP^\ast$ to estimate the approximation ratio of our
algorithms.


\paragraph{Completeness and facility splitting.}
Define $(\bfx^\ast, \bfy^\ast)$ to be \emph{complete} if
$x_{ij}^\ast>0$ implies that $x_{ij}^\ast=y_i^\ast$ for all $i,j$. In
other words, each connection either uses a site fully or not at all.
As shown by Chudak and Shmoys~\cite{ChudakS04}, we can modify the
given instance by adding at most $|\clientset|$ sites to obtain an
equivalent instance that has a complete optimal solution, where
``equivalent" means that the values of $F^\ast$, $C^\ast$ and
$\LP^\ast$, as well as $\OPT$, are not affected. Roughly, the argument
is this: We notice that, without loss of generality, for each client
$k$ there exists at most one site $i$ such that $0 < x_{ik}^\ast <
y_i^\ast$.  We can then perform the following \emph{facility
  splitting} operation on $i$: introduce a new site $i'$, let
$y^\ast_{i'} = y^\ast_i - x^\ast_{ik}$, redefine $y^\ast_i$ to be
$x^\ast_{ik}$, and then for each client $j$ redistribute $x^\ast_{ij}$
so that $i$ retains as much connection value as possible and $i'$
receives the rest. Specifically, we set
\begin{align*}
  &y^\ast_{i'} \;\assign\; y^\ast_i - x^\ast_{ik},\;   y^\ast_{i} \;\assign\; x^\ast_{ik}, \quad \text{ and }\\
  &x^\ast_{i'j} \;\assign\;\max( x^\ast_{ij} - x^\ast_{ik}, 0 ),\;	 x^\ast_{ij} \;\assign\; \min( x^\ast_{ij} , x^\ast_{ik}) 
			\quad	\textrm{for all}\ j \neq k.
\end{align*}
This operation eliminates the partial connection between $k$
and $i$ and does not create any new partial
connections. Each client can split at most one site and
hence we shall have at most $|\clientset|$ more sites.

By the above paragraph,  without loss of generality we can
assume that the optimal fractional solution $(\bfx^\ast, \bfy^\ast)$
is complete. This assumption will in fact greatly simplify some of
the arguments in the paper. Additionally, we will frequently use the facility
splitting operation described above in our algorithms to obtain fractional solutions with
desirable properties.

\section{Reduction to Polynomial Demands}
\label{sec: polynomial demands}

This section presents a \emph{demand reduction} trick that
reduces the problem for arbitrary demands to a special case
where demands are bounded by $|\sitesset|$, the number of
sites.  (The formal statement is a little more technical --
see Theorem~\ref{thm: reduction to polynomial}.)  Our
algorithms in the sections that follow process individual
demands of each client one by one, and thus they critically
rely on the demands being bounded polynomially in terms of
$|\sitesset|$ and $|\clientset|$ to keep the overall running time polynomial.

The reduction is based on an optimal fractional solution
$(\bfx^\ast,\bfy^\ast)$ of LP~(\ref{eqn:fac_primal}). From the
optimality of this solution, we can also assume that
$\sum_{i\in\sitesset} x^\ast_{ij} = r_j$ for all
$j\in\clientset$.  As explained in Section~\ref{sec: the lp
  formulation}, we can assume that $(\bfx^\ast,\bfy^\ast)$
is complete, that is $x^\ast_{ij} > 0$ implies $x^\ast_{ij}
= y^\ast_i$ for all $i,j$.  We split this solution into two
parts, namely $(\bfx^\ast,\bfy^\ast) = (\hatbfx,\hatbfy)+
(\dotbfx,\dotbfy)$, where
\begin{align*}
\haty_i &\;\assign\; \floor{y_i^\ast}, \quad
			\hatx_{ij} \;\assign\; \floor{x_{ij}^\ast} \quad\textrm{and}
			\\
\doty_i &\;\assign\; y_i^\ast - \floor{y_i^\ast}, \quad
 	\dotx_{ij} \;\assign\; x_{ij}^\ast -  \floor{x_{ij}^\ast}
\end{align*}
for all $i,j$. Now we construct two
FTFP instances $\hatcalI$ and $\dotcalI$ with the same
parameters as the original instance, except that the demand of each client $j$ is
$\hatr_j = \sum_{i\in\sitesset} \hatx_{ij}$ in instance $\hatcalI$ and
$\dotr_j = \sum_{i\in\sitesset} \dotx_{ij} = r_j - \hatr_j$ in instance $\dotcalI$. 
It is obvious that if we have integral solutions to both $\hatcalI$
and $\dotcalI$ then, when added together, they form an integral
solution to the original instance.  Moreover, we have the
following lemma.


\begin{lemma}\label{lem: polynomial demands partition}
{\rm (i)}
  $(\hatbfx, \hatbfy)$ is a feasible integral solution to
  instance $\hatcalI$.

\noindent
{\rm (ii)}
  $(\dotbfx, \dotbfy)$ is a feasible fractional
  solution to instance $\dotcalI$.

\noindent
{\rm (iii)}
$\dotr_j\leq |\sitesset|$ for every client $j$.

\end{lemma}

\begin{proof}
(i) For feasibility, we need to verify that the constraints of LP~(\ref{eqn:fac_primal})
are satisfied. Directly from the definition, we have $\hatr_j = \sum_{i\in\sitesset} \hatx_{ij}$.
For any $i$ and $j$, by the feasibility of $(\bfx^\ast,\bfy^\ast)$ we have
$\hatx_{ij} = \floor{x_{ij}^\ast} \le \floor{y^\ast_i} = \haty_i$.

(ii) From the definition, we have  $\dotr_j = \sum_{i\in\sitesset} \dotx_{ij}$.
It remains to show that $\doty_i \geq \dotx_{ij}$ for all $i,j$. 
If $x_{ij}^\ast=0$, then $\dotx_{ij}=0$ and we are done. 
Otherwise, by completeness, we have $x_{ij}^\ast=y_i^\ast$. 
Then  $\doty_i = y_i^\ast - \floor{y_i^\ast} = x_{ij}^\ast - \floor{x_{ij}^\ast} =\dotx_{ij}$. 

(iii) From the definition of $\dotx_{ij}$ we have
  $\dotx_{ij} < 1$.  Then the bound follows from the definition of $\dotr_j$.
\end{proof}

Notice that our construction relies on the completeness assumption; in fact, it is
easy to give an example where $(\dotbfx, \dotbfy)$ would not be feasible if we
used a non-complete optimal solution $(\bfx^\ast,\bfy^\ast)$.
Note also that the solutions $(\hatbfx,\hatbfy)$ and $(\dotbfx, \dotbfy)$ are in fact
optimal for their corresponding instances, for if a better solution to $\hatcalI$ or
$\dotcalI$ existed, it could
give us a solution to $\calI$ with a smaller objective value.


\begin{theorem}\label{thm: reduction to polynomial}
  Suppose that there is a polynomial-time algorithm $\calA$
  that, for any instance of {\FTFP} with maximum demand
  bounded by $|\sitesset|$, computes an integral solution
  that approximates the fractional optimum of this instance
  within factor $\rho\geq 1$.  Then there is a
  $\rho$-approximation algorithm $\calA'$ for {\FTFP}.
\end{theorem}


\begin{proof}
  Given an {\FTFP} instance with arbitrary demands, Algorithm~$\calA'$ works
as follows: it solves the LP~(\ref{eqn:fac_primal}) to obtain a
  fractional optimal solution $(\bfx^\ast,\bfy^\ast)$, then it constructs
  instances $\hatcalI$ and $\dotcalI$ described above,  applies
  algorithm~$\calA$ to $\dotcalI$, and finally combines (by adding
  the values) the integral solution $(\hatbfx, \hatbfy)$ of
  $\hatcalI$ and the integral solution of $\dotcalI$ produced
  by $\calA$. This clearly produces a feasible integral
  solution for the original instance $\calI$.
The solution produced by $\calA$ has cost at most
$\rho\cdot\cost(\dotbfx,\dotbfy)$, because $(\dotbfx,\dotbfy)$
is feasible for $\dotcalI$. Thus the cost of $\calA'$ is at most
 \begin{align*}
 \cost(\hatbfx, \hatbfy) + \rho\cdot\cost(\dotbfx,\dotbfy)
	\le
 \rho(\cost(\hatbfx, \hatbfy) + \cost(\dotbfx,\dotbfy))
		= \rho\cdot\LP^\ast \le \rho\cdot\OPT,
  \end{align*}
where the first inequality follows from $\rho\geq 1$. This completes
the proof.
\end{proof}

\section{Adaptive Partitioning}
\label{sec: adaptive partitioning}

In this section we develop our second technique, which we
call \emph{adaptive partitioning}. Given an FTFP instance
and an optimal fractional solution $(\bfx^\ast, \bfy^\ast)$
to LP~(\ref{eqn:fac_primal}), we split each client $j$ into
$r_j$ individual \emph{unit demand points} (or just
\emph{demands}), and we split each site $i$ into no more
than $|\sitesset|+2R|\clientset|^2$ \emph{facility points} (or
\emph{facilities}), where $R=\max_{j\in\clientset}r_j$. We
denote the demand set by $\demandset$ and the facility set
by $\facilityset$, respectively.  We will also partition
$(\bfx^\ast,\bfy^\ast)$ into a fractional solution
$(\barbfx,\barbfy)$ for the split instance.  We will
typically use symbols $\nu$ and $\mu$ to index demands and
facilities respectively, that is $\barbfx =
(\barx_{\mu\nu})$ and $\barbfy = (\bary_{\mu})$.  As before,
the \emph{neighborhood of a demand} $\nu$ is
$\wbarN(\nu)=\braced{\mu\in\facilityset \suchthat
  \barx_{\mu\nu}>0}$.  We will use notation $\nu\in j$ to
mean that $\nu$ is a demand of client $j$; similarly,
$\mu\in i$ means that facility $\mu$ is on site
$i$. Different demands of the same client (that is,
$\nu,\nu'\in j$) are called \emph{siblings}.  Further, we
use the convention that $f_\mu = f_i$ for $\mu\in i$,
$\alpha_\nu^\ast = \alpha_j^\ast$ for $\nu\in j$ and
$d_{\mu\nu} = d_{\mu j} = d_{ij}$ for $\mu\in i$ and $\nu\in
j$.  We define $\concost_{\nu}
=\sum_{\mu\in\wbarN(\nu)}d_{\mu\nu}\barx_{\mu\nu} =
\sum_{\mu\in\facilityset}d_{\mu\nu}\barx_{\mu\nu}$. 
One can think of $\concost_{\nu}$ as the
average connection cost of demand $\nu$, if we chose a
connection to facility $\mu$ with probability
$\barx_{\mu\nu}$. In our partitioned fractional solution we
guarantee for every $\nu$ that $\sum_{\mu\in\facilityset}
\barx_{\mu\nu}=1$.

Some demands in $\demandset$ will be designated as
\emph{primary demands} and the set of primary demands will
be denoted by $P$. By definition we have $P\subseteq \demandset$.
 In addition, we will use the overlap
structure between demand neighborhoods to define a mapping
that assigns each demand $\nu\in\demandset$ to some primary
demand $\kappa\in P$. As shown in the rounding algorithms in
later sections, for each primary demand we guarantee exactly
one open facility in its neighborhood, while for a
non-primary demand, there is constant probability that none
of its neighbors open. In this case we estimate its
connection cost by the distance to the facility opened in
its assigned primary demand's neighborhood. For this reason
the connection cost of a primary demand must be ``small''
compared to the non-primary demands assigned to it. We also
need sibling demands assigned to different primary demands to satisfy
the fault-tolerance requirement. Specifically, this
partitioning will be constructed to satisfy a number of
properties that are detailed below.
\begin{description}
	
      \renewcommand{\theenumii}{(\alph{enumii})}
      \renewcommand{\labelenumii}{\theenumii}

\item{(PS)} \emph{Partitioned solution}.
Vector $(\barbfx,\barbfy)$ is a partition of $(\bfx^\ast,\bfy^\ast)$, with unit-value
  demands, that is:

	\begin{enumerate}
	\item \label{PS:one} 
          $\sum_{\mu\in\facilityset} \barx_{\mu\nu} = 1$ for each demand $\nu\in\demandset$. 
	\item \label{PS:xij} $\sum_{\mu\in i, \nu\in j} \barx_{\mu\nu}
          = x^\ast_{ij}$ for each site $i\in\sitesset$ and client $j\in\clientset$.
	\item \label{PS:yi}
          $\sum_{\mu\in i} \bary_{\mu} = y^\ast_i$ for each site $i\in\sitesset$.
	\end{enumerate}
		
\item{(CO)} \emph{Completeness.}
	Solution   $(\barbfx,\barbfy)$ is complete, that is $\barx_{\mu\nu}\neq 0$ implies
				$\barx_{\mu\nu} = \bary_{\mu}$, for all $\mu\in\facilityset, \nu\in\demandset$.

\item{(PD)} \emph{Primary demands.}
	Primary demands satisfy the following conditions:

	\begin{enumerate}
		
	\item\label{PD:disjoint}  For any two different primary demands $\kappa,\kappa'\in P$ we have
				$\wbarN(\kappa)\cap \wbarN(\kappa') = \emptyset$.

	\item \label{PD:yi} For each site $i\in\sitesset$, 
		$ \sum_{\mu\in i}\sum_{\kappa\in P}\barx_{\mu\kappa} \leq y_i^\ast$.
		
	\item \label{PD:assign} Each demand $\nu\in\demandset$ is assigned
        to one primary demand $\kappa\in P$ such that

  			\begin{enumerate}
	
				\item \label{PD:assign:overlap} $\wbarN(\nu) \cap \wbarN(\kappa) \neq \emptyset$, and
				\item \label{PD:assign:cost} $\concost_{\nu}+\alpha_{\nu}^\ast \geq
        			\concost_{\kappa}+\alpha_{\kappa}^\ast$.

			\end{enumerate}

	\end{enumerate}
	
\item{(SI)} \emph{Siblings}. For any pair $\nu,\nu'$ of different siblings we have
  \begin{enumerate}

	\item \label{SI:siblings disjoint}
		  $\wbarN(\nu)\cap \wbarN(\nu') = \emptyset$.
		
	\item \label{SI:primary disjoint} If $\nu$ is assigned to a primary demand $\kappa$ then
 		$\wbarN(\nu')\cap \wbarN(\kappa) = \emptyset$. In particular, by Property~(PD.\ref{PD:assign:overlap}),
		this implies that different sibling demands are assigned to different primary demands.

	\end{enumerate}
	
\end{description}

As we shall demonstrate in later sections, these properties allow us
to extend known UFL rounding algorithms to obtain an integral solution
to our FTFP problem with a matching approximation ratio. Our
partitioning is ``adaptive" in the sense that it is constructed one
demand at a time, and the connection values for the demands of a
client depend on the choice of earlier demands, of this or other
clients, and their connection values. We would like to point out that
the adaptive partitioning process for the $1.575$-approximation
algorithm (Section~\ref{sec: 1.575-approximation}) is more subtle than that for 
the $3$-apprximation (Section~\ref{sec: 3-approximation}) and the
$1.736$-approximation algorithms (Section~\ref{sec:
  1.736-approximation}), due to the introduction of close and far
neighborhood.


\paragraph{Implementation of Adaptive Partitioning.}
We now describe an algorithm for partitioning the instance
and the fractional solution so that the properties (PS),
(CO), (PD), and (SI) are satisfied.  Recall that
$\facilityset$ and $\demandset$, respectively, denote the
sets of facilities and demands that will be created in this
stage, and $(\barbfx,\barbfy)$ is the partitioned solution
to be computed. 

The adaptive partitioning algorithm consists of two phases:
Phase 1 is called the partitioning phase and Phase 2 is called
the augmenting phase. Phase 1 is done in iterations, where
in each iteration we find the ``best'' client $j$ and create a
new demand $\nu$ out of it. This demand either becomes a
primary demand itself, or it is assigned to some existing
primary demand. We call a client $j$ \emph{exhausted} when
all its $r_j$ demands have been created and assigned to some
primary demands. Phase 1 completes when all clients are
exhausted. In Phase 2 we ensure that every demand has a
total connection values $\barx_{\mu\nu}$ equal to $1$, that is condition (PS.\ref{PS:one}).

For each site $i$ we will initially create one ``big" facility $\mu$
with initial value $\bary_\mu = y^\ast_i$.  While we partition the
instance, creating new demands and connections, this facility may end
up being split into more facilities to preserve completeness of the
fractional solution. Also, we will gradually decrease the fractional
connection vector for each client $j$, to account for the demands
already created for $j$ and their connection values.  These decreased
connection values will be stored in an auxiliary vector
$\tildebfx$. The intuition is that $\tildebfx$ represents the part of
$\bfx^\ast$ that still has not been allocated to existing demands and
future demands can use $\tildebfx$ for their connections. For
technical reasons, $\tildebfx$ will be indexed by facilities (rather
than sites) and clients, that is $\tildebfx = (\tildex_{\mu j})$.  At
the beginning, we set $\tildex_{\mu j}\assign x_{ij}^\ast$ for each
$j\in\clientset$, where $\mu\in i$ is the single facility created
initially at site $i$.  At each step, whenever we create a new demand
$\nu$ for a client $j$, we will define its values $\barx_{\mu\nu}$ and
appropriately reduce the values $\tildex_{\mu j}$, for all facilities
$\mu$. We will deal with two types of neighborhoods, with respect to
$\tildebfx$ and $\barbfx$, that is $\wtildeN(j)=\{\mu\in\facilityset
\suchthat\tildex_{\mu j} > 0\}$ for $j\in\clientset$ and
$\wbarN(\nu)=\{\mu\in\facilityset \suchthat \barx_{\mu\nu} >0\}$ for
$\nu\in\demandset$.  During this process we preserve the completeness
(CO) of the fractional solutions $\tildebfx$ and $\barbfx$. More
precisely, the following properties will hold for every facility $\mu$
after every iteration:
\begin{description}
	
	\item{(c1)} For each demand $\nu$ either $\barx_{\mu\nu}=0$ or
			$\barx_{\mu\nu}=\bary_{\mu}$. This is the same
      condition as condition (CO), yet we repeat it here as
      (c1) needs to hold after every iteration, while
      condition (CO) only applies to the final partitioned
      fractional solution $(\barbfx, \barbfy)$.

	\item{(c2)} For each client $j$,
			either $\tildex_{\mu j}=0$ or $\tildex_{\mu j}=\bary_{\mu}$.
			
\end{description}

A full description of the algorithm is given in
Pseudocode~\ref{alg:lpr2}.  Initially, the set $U$ of
non-exhausted clients contains all clients, the set
$\demandset$ of demands is empty, the set $\facilityset$ of
facilities consists of one facility $\mu$ on each site $i$
with $\bary_\mu = y^\ast_i$, and the set $P$ of primary
demands is empty (Lines 1--4).  In one iteration of the
while loop (Lines 5--8), for each client $j$ we
compute a quantity called $\tcc(j)$ (tentative connection
cost), that represents the average distance from $j$ to the
set $\wtildeN_1(j)$ of the nearest facilities $\mu$ whose
total connection value to $j$ (the sum of $\tildex_{\mu
  j}$'s) equals $1$.  This set is computed by Procedure
$\NearestUnitChunk()$ (see Pseudocode~\ref{alg:helper},
Lines~1--9), which adds facilities to $\wtildeN_1(j)$ in
order of nondecreasing distance, until the total connection
value is exactly $1$. (The procedure actually uses the
$\bary_\mu$ values, which are equal to the connection values,
by the completeness condition (c2).)  This may require splitting the last added
facility and adjusting the connection values so that
conditions (c1) and (c2) are preserved.


\begin{algorithm}[ht]
  \caption{Algorithm: Adaptive Partitioning}
  \label{alg:lpr2}
  \begin{algorithmic}[1]
    \Require $\sitesset$, $\clientset$, $(\bfx^\ast,\bfy^\ast)$
    \Ensure  $\facilityset$,  $\demandset$, $(\barbfx, \barbfy)$ 
    \Comment Unspecified $\barx_{\mu \nu}$'s and $\tildex_{\mu j}$'s are assumed to be $0$

    \State $\tildebfr \assign \bfr, U\assign \clientset, \facilityset\assign \emptyset,
    \demandset\assign \emptyset, P\assign \emptyset$
    \Comment{Phase 1}

    \For{each site $i\in\sitesset$} 
    \State create a facility $\mu$ at $i$ and add $\mu$ to $\facilityset$
    \State $\bary_\mu \assign y_i^\ast$ and $\tildex_{\mu j}\assign
    x_{ij}^\ast$ for each $j\in\clientset$ 
    \EndFor

    \While{$U\neq \emptyset$}
    \For{each $j\in U$}
    \State $\wtildeN_1(j) \assign {\NearestUnitChunk}(j, \facilityset, \tildebfx, \barbfx, \barbfy)$ \Comment see Pseudocode~\ref{alg:helper}
    \State $\tcc(j)\assign \sum_{\mu\in \wtildeN_1(j)} d_{{\mu}j}\cdot \tildex_{\mu j}$
    \EndFor
 
    \State $p \assign {\argmin}_{j\in U}\{ \tcc(j)+\alpha_j^\ast \}$
    \State create a new demand $\nu$ for client $p$

    \If{$\wtildeN_1 (p)\cap \wbarN(\kappa) \neq \emptyset$
      for some primary demand $\kappa\in P$}
    \State assign $\nu$ to $\kappa$
    \State $\barx_{\mu \nu}\assign \tildex_{\mu p}$ and $\tildex_{\mu p}\assign 0$ for each $\mu \in \wtildeN(p) \cap \wbarN(\kappa)$
    \Else 
    \State make $\nu$ primary, $P \assign P \cup \{\nu\}$, assign $\nu$ to itself
    \State set $\barx_{\mu\nu} \assign \tildex_{\mu p}$ and $\tildex_{\mu p}\assign 0$ for each $\mu\in \wtildeN_1(p)$

    \EndIf
    \State $\demandset\assign \demandset\cup \{\nu\},
    \tilder_p \assign \tilder_p -1$
	\State \textbf{if} {$\tilder_p=0$} \textbf{then} $U\assign U \setminus \{p\}$
    \EndWhile

    \For{each client $j\in\clientset$} \Comment{Phase 2}
    \For{each demand $\nu\in j$}    \Comment{each client $j$ has $r_j$ demands}
    \State \textbf{if} $\sum_{\mu\in \wbarN(\nu)}\barx_{\mu\nu}<1$
    \textbf{then} $\AugmentToUnit(\nu, j, \facilityset, \tildebfx, \barbfx, \barbfy)$ \Comment see Pseudocode~\ref{alg:helper}
    \EndFor
    \EndFor
  \end{algorithmic}
\end{algorithm}
\begin{algorithm}[ht]
  \caption{Helper functions used in Pseudocode~\ref{alg:lpr2}}
  \label{alg:helper}
  \begin{algorithmic}[1]
    \Function{\NearestUnitChunk}{$j, \facilityset, \tildebfx, \barbfx,\barbfy$}		
						\Comment upon return, $\sum_{\mu\in\wtildeN_1(j)} \tildex_{\mu j} = 1$
    \State Let $\wtildeN(j) = \{\mu_1,...,\mu_{q}\}$ where $d_{\mu_1 j} \leq d_{\mu_2 j} \leq \ldots \leq d_{\mu_{q j}}$
    \State Let $l$ be such that $\sum_{k=1}^{l} \bary_{\mu_k} \geq 1$ and $\sum_{k=1}^{l -1} \bary_{\mu_{k}} < 1$
    \State Create a new facility $\sigma$ at the same site as $\mu_l$ and add it to $\facilityset$
			\Comment split $\mu_l$
    \State Set $\bary_{\sigma}\assign \sum_{k=1}^{l} \bary_{\mu_{k}}-1$
					and $\bary_{\mu_l} \assign \bary_{\mu_l} - \bary_{\sigma}$
    \State For each $\nu\in\demandset$ with $\barx_{\mu_{l}\nu}>0$
 			set $\barx_{\mu_{l}\nu} \assign \bary_{\mu_l}$ and $\barx_{\sigma \nu} \assign \bary_{\sigma}$
    \State For each $j'\in\clientset$ with $\tildex_{\mu_{l} j'}>0$ (including $j$)
			set $\tildex_{\mu_l j'} \assign \bary_{\mu_l}$ and $\tildex_{\sigma j'} \assign \bary_\sigma$
	\State (All other new connection values are set to $0$)
    \State \Return $\wtildeN_1(j) = \{\mu_{1},\ldots,\mu_{l-1}, \mu_{l}\}$    				
    \EndFunction

    \Function{\AugmentToUnit}{$\nu, j, \facilityset, \tildebfx, \barbfx, \barbfy$}
    					\Comment $\nu$ is a demand of client $j$
    \While{$\sum_{\mu\in \facilityset} \barx_{\mu\nu} <1$}
    					\Comment upon return, $\sum_{\mu\in\wbarN(\nu)} \barx_{\mu\nu} = 1$
    \State Let $\eta$ be any facility such that $\tildex_{\eta j} > 0$
    \If{$1-\sum_{\mu\in \facilityset} \barx_{\mu\nu} \geq \tildex_{\eta j}$}
    \State $\barx_{\eta\nu} \assign \tildex_{\eta j}, \tildex_{\eta j} \assign 0$
    \Else
    \State Create a new facility $\sigma$ at the same site as $\eta$ and add it to $\facilityset$
    					\Comment split $\eta$
    \State Let $\bary_\sigma \assign 1-\sum_{\mu\in \facilityset} \barx_{\mu\nu}, \bary_{\eta} \assign \bary_{\eta} - \bary_{\sigma}$
    \State Set $\barx_{\sigma\nu}\assign \bary_{\sigma},\; \barx_{\eta \nu} \assign  0,\; \tildex_{\eta j} \assign \bary_{\eta}, \; \tildex_{\sigma j} \assign 0$
    \State For each $\nu' \neq \nu$ with $\barx_{\eta \nu'}>0$, set $\barx_{\eta \nu'} \assign \bary_{\eta},\; \barx_{\sigma \nu'} \assign \bary_{\sigma}$
    \State For each $j' \neq j$ with $\tildex_{\eta j'}>0$, set $\tildex_{\eta j'} \assign \bary_{\eta}, \tildex_{\sigma j'} \assign \bary_{\sigma}$
	\State  (All other new connection values are set to $0$)
    \EndIf
    \EndWhile
    \EndFunction
  \end{algorithmic}
\end{algorithm}


The next step is to pick a client $p$ with minimum
$\tcc(p)+\alpha_p^\ast$ and create a demand $\nu$ for $p$
(Lines~9--10). If $\wtildeN_1(p)$ overlaps the neighborhood
of some existing primary demand $\kappa$ (if there are
multiple such $\kappa$'s, pick any of them), we assign $\nu$
to $\kappa$, and $\nu$ acquires all the connection values
$\tildex_{\mu p}$ between client $p$ and facility $\mu$ in
$\wtildeN(p)\cap \wbarN(\kappa)$ (Lines~11--13). Note that
although we check for overlap with $\wtildeN_1(p)$, we then
move all facilities in the intersection with $\wtildeN(p)$,
a bigger set, into $\wbarN(\nu)$.  The other case is when
$\wtildeN_1(p)$ is disjoint from the neighborhoods of all
existing primary demands. Then, in Lines~15--16, $\nu$
becomes itself a primary demand and we assign $\nu$ to
itself. It also inherits the connection values to all
facilities $\mu\in\wtildeN_1(p)$ from $p$ (recall that
$\tildex_{\mu p} = \bary_{\mu}$), with all other
$\barx_{\mu\nu}$ values set to $0$.

At this point all primary demands satisfy
Property~(PS.\ref{PS:one}), but this may not be true for
non-primary demands. For those demands we still may need to
adjust the $\barx_{\mu\nu}$ values so that the total
connection value for $\nu$, that is $\connsum(\nu) \stackrel{\mathrm{def}}{=}
\sum_{\mu\in\facilityset}\barx_{\mu \nu}$, is equal $1$. This
is accomplished by Procedure $\AugmentToUnit()$ (definition
in Pseudocode~\ref{alg:helper}, Lines~10--21) that allocates
to $\nu\in j$ some of the remaining connection values
$\tildex_{\mu j}$ of client $j$ (Lines 19--21).
$\AugmentToUnit()$ will repeatedly pick any facility $\eta$ with
$\tildex_{\eta j} >0$.  If $\tildex_{\eta j} \leq
1-\connsum(\nu)$, then the connection value $\tildex_{\eta
  j}$ is reassigned to $\nu$. 
Otherwise, $\tildex_{\eta j} >
1-\connsum(\nu)$, in which case we split $\eta$ so that
connecting $\nu$ to one of the created copies of $\eta$ will
make $\connsum(\nu)$ equal $1$, and we'll be done.

\smallskip

Notice that we start with $|\sitesset|$ facilities and in
each iteration of the while loop in Line~5 (Pseudocode~\ref{alg:lpr2}) each client causes at most one split.
 We have a total of no more than $R|\clientset|$ iterations as in
each iteration we create one demand. (Recall that $R =
\max_jr_j$.) In Phase 2 we do an augment step for each
demand $\nu$ and this creates no more than $R|\clientset|$
new facilities.  So the total number of facilities we
created will be at most $|\sitesset|+ R|\clientset|^2 +
R|\clientset| \leq |\sitesset| + 2R|\clientset|^2$, which is
polynomial in $|\sitesset|+|\clientset|$ due to our earlier
bound on $R$.

\paragraph{Example.}
We now illustrate our partitioning algorithm with an example, where the FTFP instance
has four sites and four clients. The demands are $r_1=1$ and $r_2=r_3=r_4=2$.
The facility costs are $f_i = 1$ for all $i$. The distances are defined as follows: 
$d_{ii} = 3$ for $i=1,2,3,4$ and $d_{ij} = 1$ for all $i\neq j$. 
Solving the LP(\ref{eqn:fac_primal}), we obtain the fractional solution given in
Table~\ref{tbl:example_opt}.
{
\small
\begin{table}[ht]
\hfill
\setlength{\extrarowheight}{4pt}
\begin{subtable}{0.2\textwidth}
  \centering
  \begin{tabular}{c | c c c c | c }
    $x_{ij}^\ast$ & $1$ & $2$ & $3$ & $4$ & $y_{i}^\ast$\\
    \hline
    $1$ & 0 & $\fourthirds$ & $\fourthirds$ & $\fourthirds$ & $\fourthirds$ \\
    $2$ & $\onethird$ & 0 & $\onethird$ & $\onethird$ & $\onethird$ \\
    $3$ & $\onethird$ & $\onethird$ & 0 & $\onethird$ & $\onethird$ \\
    $4$ & $\onethird$ & $\onethird$ & $\onethird$ & 0 & $\onethird$ \\
  \end{tabular}
  \subcaption{}
  \label{tbl:example_opt}
\end{subtable}
\hspace{0.8in}
\begin{subtable}{0.4\textwidth}
  \centering
  \begin{tabular}{c | c c c c c c c | c} 
    $\barx_{\mu\nu}$ & $1'$ & $2'$ & $2''$ & $3'$ & $3''$ & $4'$ & $4''$ & $\bary_{\mu}$ \\
    \hline
    $\dot{1}$ & 0 & 1 & 0 & 1 & 0 & 1 & 0 & 1\\
    $\ddot{1}$ & 0 & 0 & $\onethird$ & 0 & $\onethird$ & 0 & $\onethird$ & $\onethird$ \\
    $\dot{2}$ & $\onethird$ & 0 & 0 & 0 & $\onethird$ & 0 & $\onethird$  & $\onethird$ \\
    $\dot{3}$ & $\onethird$ & 0 & $\onethird$ & 0 & 0 & 0 & $\onethird$  & $\onethird$ \\
    $\dot{4}$ & $\onethird$ & 0 & $\onethird$ & 0 & $\onethird$ & 0 & 0 & $\onethird$ \\
  \end{tabular}
  \subcaption{}
  \label{tbl:example_part}
\end{subtable}
\hfill{\ }
\caption{
  An example of an execution of the partitioning algorithm.
  (a) An optimal fractional solution $x^\ast,y^\ast$.
  (b) The partitioned solution. $j'$ and $j''$ denote the first and second demand of a client $j$, 
	and $\dot{\imath}$ and $\ddot{\imath}$ denote the first and second facility at site $i$.}
\end{table}
}

It is easily seen that the fractional solution in
Table~\ref{tbl:example_opt} is optimal and complete ($x_{ij}^\ast > 0$
implies $x_{ij}^\ast = y_i^\ast$). The dual optimal solution has all
$\alpha_j^\ast = 4/3$ for $j=1,2,3,4$.

Now we perform Phase 1, the adaptive partitioning, following the
description in Pseudocode~\ref{alg:lpr2}. To streamline the
presentation, we assume that all ties are broken in favor of
lower-numbered clients, demands or facilities.  First we create one
facility at each of the four sites, denoted as $\dot{1}$, $\dot{2}$,
$\dot{3}$ and $\dot{4}$ (Line~2--4, Pseudocode~\ref{alg:lpr2}).  We
then execute the ``while'' loop in Line 5
Pseudocode~\ref{alg:lpr2}. This loop will have seven iterations.
Consider the first iteration. In Line 7--8 we compute $\tcc(j)$ for
each client $j=1,2,3,4$ in $U$. When computing $\wtildeN_1(2)$,
facility $\dot{1}$ will get split into $\dot{1}$ and $\ddot{1}$ with
$\bary_{\dot{1}}=1$ and $\bary_{\ddot{1}} = 1/3$. (This will happen in
Line~4--7 of Pseudocode~\ref{alg:helper}.)  Then, in Line~9 we will
pick client $p=1$ and create a demand denoted as $1'$ (see
Table~\ref{tbl:example_part}). Since there are no primary demands yet,
we make $1'$ a primary demand with $\wbarN(1') = \wtildeN_1(1) =
\{\dot{2}, \dot{3}, \dot{4}\}$. Notice that client $1$ is exhausted
after this iteration and $U$ becomes $\{2,3,4\}$.

In the second iteration we compute $\tcc(j)$ for $j=2,3,4$ and pick
client $p=2$, from which we create a new demand $2'$. We have
$\wtildeN_1(2) = \{\dot{1}\}$, which is disjoint from $\wbarN(1')$. So
we create a demand $2'$ and make it primary, and set $\wbarN(2') =
\{\dot{1}\}$. In the third iteration we compute $\tcc(j)$ for
$j=2,3,4$ and again we pick client $p=2$. Since $\wtildeN_1(2) =
\{\ddot{1}, \dot{3}, \dot{4}\}$ overlaps with $\wbarN(1')$, we create
a demand $2''$ and assign it to $1'$. We also set $\wbarN(2'') =
\wbarN(1') \cap \wtildeN(2) = \{\dot{3}, \dot{4}\}$. After this
iteration client $2$ is exhausted and we have $U = \{3,4\}$.

In the fourth iteration we compute $\tcc(j)$ for client $j=3,4$. We
pick $p=3$ and create demand $3'$. Since $\wtildeN_1(3) = \{\dot{1}\}$
overlaps $\wbarN(2')$, we assign $3'$ to $2'$ and set
$\wbarN(3') = \{\dot{1}\}$. In the fifth iteration we compute
$\tcc(j)$ for client $j=3,4$ and pick $p=3$ again. At this time
$\wtildeN_1(3) = \{\ddot{1},\dot{2},\dot{4}\}$, which overlaps with
$\wbarN(1')$. So we create a demand $3''$ and assign it to $1'$, as
well as set $\wbarN(3'') = \{\dot{2}, \dot{4}\}$.

In the last two iterations we will pick client $p=4$ twice and
create demands $4'$ and $4''$. For $4'$ we have $\wtildeN_1(4) =
\{\dot{1}\}$ so we assign $4'$ to $2'$ and set $\wbarN(4') =
\{\dot{1}\}$. For $4''$ we have $\wtildeN_1(4) = \{\ddot{1}, \dot{2},
\dot{3}\}$ and we assign it to $1'$, as well as set $\wbarN(4'') =
\{\dot{2}, \dot{3}\}$.

Now that all clients are exhausted we perform Phase 2, the augmenting
phase, to construct a fractional solution in which all demands have
total connection value equal to $1$.  We iterate through each of the
seven demands created, that is $1',2',2'',3',3'',4',4''$.  $1'$ and $2'$
already have neighborhoods with total connection value of $1$, so
nothing will change in the first two iterations.
$2''$ has $\dot{3},\dot{4}$ in its neighborhood, with total connection value of
$2/3$, and $\wtildeN(2) = \{\ddot{1}\}$ at this time, so we add
$\ddot{1}$ into $\wbarN(2'')$ to make $\wbarN(2'') = \{\ddot{1},
\dot{3}, \dot{4}\}$ and now $2''$ has total connection value of
$1$. Similarly, $3''$ and $4''$ each get $\ddot{1}$ added to their
neighborhood and end up with total connection value of $1$. The other
two demands, namely $3'$ and $4'$, each have $\dot{1}$ in its
neighborhood so each of them has already its total connection value
equal $1$. This completes Phase 2.

The final partitioned fractional solution is given in
Table~\ref{tbl:example_part}. We have created a total of five
facilities $\dot{1}, \ddot{1}, \dot{2}, \dot{3}, \dot{4}$, and seven
demands, $1',2',2'',3',3'',4',4''$. It can be verified that all the
stated properties are satisfied.



\medskip

\emparagraph{Correctness.}  We now show that all the
required properties (PS), (CO), (PD) and (SI) are satisfied
by the above construction.

Properties~(PS) and (CO) follow directly from the
algorithm. (CO) is implied by the completeness condition
(c1) that the algorithm maintains after each
iteration. Condition~(PS.\ref{PS:one}) is a result of
calling Procedure~$\AugmentToUnit()$ in Line~21. To see that
(PS.\ref{PS:xij}) holds, note that
at each step the algorithm maintains the
invariant that, for every $i\in\sitesset$ and
$j\in\clientset$, we have $\sum_{\mu\in i}\sum_{\nu \in j}
\barx_{\mu \nu} + \sum_{\mu\in i} \tildex_{\mu j} =
x_{ij}^\ast$. In the end, we will create $r_j$ demands for
each client $j$, with each demand $\nu\in j$ satisfying
(PS.\ref{PS:one}), and thus $\sum_{\nu\in
  j}\sum_{\mu\in\facilityset}\barx_{\mu\nu}=r_j$.  This
implies that $\tildex_{\mu j}=0$ for every facility
$\mu\in\facilityset$, and (PS.\ref{PS:xij}) follows.
(PS.\ref{PS:yi}) holds because every time we split a
facility $\mu$ into $\mu'$ and $\mu''$, the sum of
$\bary_{\mu'}$ and $\bary_{\mu''}$ is equal to the old value of
$\bary_{\mu}$.

Now we deal with properties in group (PD).  First,
(PD.\ref{PD:disjoint}) follows directly from the algorithm,
Pseudocode~\ref{alg:lpr2} (Lines 14--16), since every
primary demand has its neighborhood fixed when created, and
that neighborhood is disjoint from those of the existing primary
demands.

Property (PD.\ref{PD:yi}) follows from (PD.\ref{PD:disjoint}), (CO) and
(PS.\ref{PS:yi}). In more detail, it can be justified as
follows. By (PD.\ref{PD:disjoint}), for each $\mu\in i$ there
is at most one $\kappa\in P$ with $\barx_{\mu\kappa} > 0$
and we have $\barx_{\mu\kappa} = \bary_{\mu}$ due do (CO).
Let $K\subseteq i$ be the set of those $\mu$'s for which
such $\kappa\in P$ exists, and denote this $\kappa$ by
$\kappa_\mu$. Then, using conditions (CO) and
(PS.\ref{PS:yi}), we have $ \sum_{\mu\in i}\sum_{\kappa\in
  P}\barx_{\mu\kappa} = \sum_{\mu\in K}\barx_{\mu\kappa_\mu}
= \sum_{\mu\in K}\bary_{\mu} \leq \sum_{\mu\in i}
\bary_{\mu} = y_i^\ast$.

Property (PD.\ref{PD:assign:overlap}) follows from the way the algorithm
assigns primary demands.  When demand $\nu$ of
client $p$ is assigned to a primary demand $\kappa$ in
Lines~11--13 of Pseudocode~\ref{alg:lpr2}, we move all
facilities in $\wtildeN(p)\cap \wbarN(\kappa)$ (the
intersection is nonempty) into $\wbarN(\nu)$, and we never
remove a facility from $\wbarN(\nu)$.  We postpone the proof 
for (PD.\ref{PD:assign:cost}) to Lemma~\ref{lem: PD:assign:cost holds}.

Finally we argue that the properties in group (SI)
hold. (SI.\ref{SI:siblings disjoint}) is easy, since for any client
$j$, each facility $\mu$ is added to the neighborhood of at most one
demand $\nu\in j$, by setting $\barx_{\mu\nu}$ to $\bary_\mu$, while
other siblings $\nu'$ of $\nu$ have $\barx_{\mu\nu'}=0$. Note that
right after a demand $\nu\in p$ is created, its neighborhood is
disjoint from the neighborhood of $p$, that is $\wbarN(\nu)\cap
\wtildeN(p) = \emptyset$, by Lines~11--13 of the algorithm. Thus all
demands of $p$ created later will have neighborhoods disjoint from the
set $\wbarN(\nu)$ before the augmenting phase 2. Furthermore,
Procedure~$\AugmentToUnit()$ preserves this property, because when it
adds a facility to $\wbarN(\nu)$ then it removes it from
$\wtildeN(p)$, and in case of splitting, one resulting facility is
added to $\wbarN(\nu)$ and the other to $\wtildeN(p)$. Property
(SI.\ref{SI:primary disjoint}) is shown below in Lemma~\ref{lem:
  property SI:primary disjoint holds}.

It remains to show Properties~(PD.\ref{PD:assign:cost}) and
(SI.\ref{SI:primary disjoint}). We show them in the lemmas
below, thus completing the description of our adaptive
partition process.


\begin{lemma}\label{lem: property SI:primary disjoint holds}
  Property~(SI.\ref{SI:primary disjoint}) holds after the
  Adaptive Partitioning stage.
\end{lemma}

\begin{proof}
  Let $\nu_1,\ldots,\nu_{r_j}$ be the demands of a client
  $j\in\clientset$, listed in the order of creation, and, for each
  $q=1,2,\ldots,r_j$, denote by $\kappa_q$ the primary demand that
  $\nu_q$ is assigned to. After the completion of Phase~1 of
  Pseudocode~\ref{alg:lpr2} (Lines 5--18), we have
  $\wbarN(\nu_s)\subseteq \wbarN(\kappa_s)$ for  $s=1,\ldots,r_j$. 
Since any two primary demands have disjoint
  neighborhoods, we have $\wbarN(\nu_s) \cap \wbarN(\kappa_q) =
  \emptyset$ for any $s\neq q$, that is
	Property~(SI.\ref{SI:primary disjoint}) holds right after Phase~1.

        After Phase~1 all neighborhoods $\wbarN(\kappa_s),
        s=1,\ldots,r_j$ have already been fixed and they do not change
        in Phase~2.  None of the facilities in $\wtildeN(j)$ appear in
        any of $\wbarN(\kappa_s)$ for $s=1,\ldots,r_j$, by the way we
        allocate facilities in Lines~13 and 16.  Therefore during the
        augmentation process in Phase~2, when we add facilities from
        $\wtildeN(j)$ to $\wbarN(\nu)$, for some $\nu\in j$
        (Line~19--21 of Pseudocode~\ref{alg:lpr2}), all the required
        disjointness conditions will be preserved.
\end{proof}


We need one more lemma before proving our last property
(PD.\ref{PD:assign:cost}).  For a client $j$ and a demand
$\nu$, we use notation $\tcc^{\nu}(j)$ for the value of
$\tcc(j)$ at the time when $\nu$ was created. (It is not
necessary that $\nu\in j$ but we assume that $j$ is not
exhausted at that time.)

\begin{lemma}\label{lem: tcc optimal}
  Let $\eta$ and $\nu$ be two demands, with $\eta$ created
  no later than $\nu$, and let $j\in\clientset$ be a client
  that is not exhausted when $\nu$ is created. Then we have
\begin{description}
	\item{(a)} $\mbox{\tcc}^\eta(j) \le \mbox{\tcc}^{\nu}(j)$, and 
	\item{(b)} if $\nu\in j$ then $\mbox{\tcc}^\eta(j) \le \concost_{\nu}$.
\end{description}
\end{lemma}

\begin{proof}
  We focus first on the time when demand $\eta$ is about to be created,
  right after the call to $\NearestUnitChunk()$ in
  Pseudocode~\ref{alg:lpr2}, Line~7.  Let $\wtildeN(j) =
  \{\mu_1,...,\mu_q\}$ with all facilities $\mu_s$ ordered
  according to nondecreasing distance from $j$.  Consider
  the following linear program:
\begin{alignat*}{1}
	\textrm{minimize} \quad & \sum_s d_{\mu_s j}z_s
			\\
	\textrm{subject to} \quad & \sum_s z_s  \ge 1
			\\
 	0 &\le z_s \le \tildex_{\mu_s j} \quad \textrm{for all}\ s
\end{alignat*}
  This is a fractional
  minimum knapsack covering problem (with knapsack size equal $1$) and its optimal fractional
  solution is the greedy solution, whose value is exactly
  $\tcc^\eta(j)$.  

On the other hand, we claim that
  $\tcc^{\nu}(j)$ can be thought of as the value of some feasible
  solution to this linear program, and that the same is true for $\concost_{\nu}$ if $\nu\in j$.
  Indeed, each of these
  quantities involves some later values $\tildex_{\mu j}$,
  where $\mu$ could be one of the facilities $\mu_s$ or a
  new facility obtained from splitting. For each $s$,
  however, the sum of all values $\tildex_{\mu j}$,
  over the facilities $\mu$ that were split from $\mu_s$, cannot exceed
 the value $\tildex_{\mu_s j}$ at the time when
  $\eta$ was created, because splitting facilities preserves this sum and
 creating new demands for $j$ can only decrease it.
Therefore both quantities
  $\tcc^{\nu}(j)$ and $\concost_{\nu}$ (for $\nu\in j$) correspond to some
  choice of the $z_s$ variables (adding up to $1$), and the
  lemma follows.
\end{proof}


\begin{lemma}\label{lem: PD:assign:cost holds}
Property~(PD.\ref{PD:assign:cost}) holds after the Adaptive Partitioning stage.
\end{lemma}

\begin{proof}
Suppose that demand $\nu\in j$ is assigned to some primary demand $\kappa\in p$.
Then
\begin{eqnarray*}
 \concost_{\kappa} + \alpha_{\kappa}^\ast \;=\; \tcc^{\kappa}(p) + \alpha^\ast_p
 					\;\le\; \tcc^{\kappa}(j) + \alpha^\ast_j   
					\;\le\; \concost_{\nu} + \alpha^\ast_\nu.
\end{eqnarray*}
We now justify this derivation. By definition we have
$\alpha_{\kappa}^\ast = \alpha^\ast_p$.  Further, by the
algorithm, if $\kappa$ is a primary demand of client $p$,
then $\concost_{\kappa}$ is equal to $\tcc(p)$ computed when
$\kappa$ is created, which is exactly $\tcc^{\kappa}(p)$. Thus
the first equation is true. The first inequality follows
from the choice of $p$ in Line~9 in
Pseudocode~\ref{alg:lpr2}. The last inequality holds
because $\alpha^\ast_j = \alpha^\ast_\nu$ (due to $\nu\in
j$), and because $\tcc^{\kappa}(j) \le \concost_{\nu}$, which
follows from Lemma~\ref{lem: tcc optimal}.
\end{proof}

We have thus proved that all properties (PS), (CO), (PD) and (SI) hold
for our partitioned fractional solution $(\barbfx,\barbfy)$. In the
following sections we show how to use these properties to round the
fractional solution to an approximate integral solution. For the
$3$-approximation algorithm (Section~\ref{sec: 3-approximation}) and
the $1.736$-approximation algorithm (Section~\ref{sec:
  1.736-approximation}), the first phase of the algorithm is exactly
the same partition process as described above. However, the
$1.575$-approximation algorithm (Section~\ref{sec:
  1.575-approximation}) demands a more sophisticated partitioning
process as the interplay between close and far neighborhood of sibling
demands result in more delicate properties that our partitioned
fractional solution must satisfy.

\section{Algorithm~{\EGUP} with Ratio $3$}
\label{sec: 3-approximation}

With the partitioned FTFP instance and its associated fractional
solution in place, we now begin to introduce our rounding algorithms.
The algorithm we describe in this section achieves ratio $3$. Although
this is still quite far from our best ratio $1.575$ that we derive
later, we include this algorithm in the paper to illustrate, in a
relatively simple setting, how the properties of our partitioned
fractional solution are used in rounding it to an integral solution
with cost not too far away from an optimal solution.  The rounding
approach we use here is an extension of the corresponding method for
UFL described in~\cite{gupta08}.

\paragraph{Algorithm~{\EGUP.}}
At a high level, we would open exactly one facility for each
primary demand $\kappa$, and each non-primary demand is
connected to the facility opened for the primary demand it
was assigned to.

More precisely, we apply a rounding process, guided by the
fractional values $(\bary_{\mu})$ and $(\barx_{\mu\nu})$,
that produces an integral solution. This integral solution
is obtained by choosing a subset of facilities in
$\facilityset$ to open, and for each demand in $\demandset$,
specifying an open facility that this demand will be
connected to.  For each primary demand $\kappa\in P$, we
want to open one facility $\phi(\kappa) \in
\wbarN(\kappa)$. To this end, we use randomization: for each
$\mu\in\wbarN(\kappa)$, we choose $\phi(\kappa) = \mu$ with
probability $\barx_{\mu\kappa}$, ensuring that exactly one
$\mu \in \wbarN(\kappa)$ is chosen. Note that
$\sum_{\mu\in\wbarN(\kappa)}\barx_{\mu\kappa}=1$, so this
distribution is well-defined.  We open this facility
$\phi(\kappa)$ and connect to $\phi(\kappa)$ all demands
that are assigned to $\kappa$.

In our description above, the algorithm is presented as a
randomized algorithm. It can be de-randomized using the
method of conditional expectations, which is commonly used
in approximation algorithms for facility location problems
and standard enough that presenting it here would be
redundant. Readers less familiar with this field are
recommended to consult \cite{ChudakS04}, where the method of
conditional expectations is applied in a context very
similar to ours.


\paragraph{Analysis.}
We now bound the expected facility cost and connection cost
by establishing the two lemmas below.


\begin{lemma}\label{lemma:3fac}
The expectation of facility cost $F_{\smallEGUP}$ of our solution is
  at most $F^\ast$.
\end{lemma}

\begin{proof}
  By Property~(PD.\ref{PD:disjoint}), the neighborhoods of
  primary demands are disjoint. Also, for any primary demand
  $\kappa\in P$, the probability that a facility
  $\mu\in\wbarN(\kappa)$ is chosen as the open facility
  $\phi(\kappa)$ is $\barx_{\mu\kappa}$. Hence the expected
  total facility cost is
\begin{align*}
    \Exp[F_{\smallEGUP}]
	&= \textstyle{\sum_{\kappa\in P}\sum_{\mu\in\wbarN(\kappa)}} f_{\mu} \barx_{\mu\kappa}
	\\
	&= \textstyle{\sum_{\kappa\in P}\sum_{\mu\in\facilityset}} f_{\mu} \barx_{\mu\kappa} 
	\\
	&= \textstyle{\sum_{i\in\sitesset}} f_i \textstyle{\sum_{\mu\in i}\sum_{\kappa\in P}} \barx_{\mu\kappa} 
	\\
	&\leq \textstyle{\sum_{i\in\sitesset}} f_i y_i^\ast 
	= F^\ast,
\end{align*}
where the inequality follows from Property~(PD.\ref{PD:yi}).
\end{proof}


\begin{lemma}\label{lemma:3dist}
The expectation of connection cost $C_{\smallEGUP}$ of our solution
is at most  $C^\ast+2\cdot\LP^\ast$.
\end{lemma}

\begin{proof}
  For a primary demand $\kappa$, its expected connection cost is
  $C_{\kappa}^{\avg}$ because we choose facility $\mu$ with
  probability $\barx_{\mu\kappa}$.

  Consider a non-primary demand $\nu$ assigned to a primary demand
  $\kappa\in P$. Let $\mu$ be any facility in $\wbarN(\nu) \cap
  \wbarN(\kappa)$.  Since $\mu$ is in both $\wbarN(\nu)$ and
  $\wbarN(\kappa)$, we have $d_{\mu\nu} \leq \alpha_{\nu}^\ast$ and
  $d_{\mu\kappa} \leq \alpha_{\kappa}^\ast$ (This follows from the
  complementary slackness conditions since
  $\alpha_{\nu}^\ast=\beta_{\mu\nu}^\ast + d_{\mu\nu}$ for each
  $\mu\in \wbarN(\nu)$.). Thus, applying the triangle inequality, for
  any fixed choice of facility $\phi(\kappa)$ we have
\begin{equation*}
    d_{\phi(\kappa)\nu} \leq d_{\phi(\kappa)\kappa}+d_{\mu\kappa}+d_{\mu\nu}
    \leq d_{\phi(\kappa)\kappa} + \alpha_{\kappa}^\ast + \alpha_{\nu}^\ast.
\end{equation*}
Therefore the expected distance from $\nu$ to its facility $\phi(\kappa)$ is 
\begin{align*}
  \Exp[  d_{\phi(\kappa)\nu}   ] &\le \concost_{\kappa} + \alpha_{\kappa}^\ast + \alpha_{\nu}^\ast 
\\
  &\leq \concost_{\nu} + \alpha_{\nu}^\ast + \alpha_{\nu}^\ast
   = \concost_{\nu} + 2\alpha_{\nu}^\ast,
  \end{align*}
  where the second inequality follows from Property~(PD.\ref{PD:assign:cost}).  
From the definition of $\concost_{\nu}$ and Property~(PS.\ref{PS:xij}), for any $j\in \clientset$ 
we have
\begin{align*}
\sum_{\nu\in j} \concost_{\nu} &= \sum_{\nu\in j}\sum_{\mu\in\facilityset}d_{\mu\nu}\barx_{\mu\nu}
			\\
 			&= \sum_{i\in\sitesset} d_{ij}\sum_{\nu\in j}\sum_{\mu\in i}\barx_{\mu\nu}
			\\
			&= \sum_{i\in\sitesset} d_{ij}x^\ast_{ij} 
			= C^\ast_j.
\end{align*}
Thus, summing over all demands, the expected total connection cost is
\begin{align*}
    \Exp[C_{\smallEGUP}] &\le 
			\textstyle{\sum_{j\in\clientset} \sum_{\nu\in j}} (\concost_{\nu} + 2\alpha_{\nu}^\ast) 
			\\
    	& = \textstyle{\sum_{j\in\clientset}} (C_j^\ast + 2r_j\alpha_j^\ast)
 		= C^\ast + 2\cdot\LP^\ast,
\end{align*}
completing the proof of the lemma.
\end{proof}


\begin{theorem}
Algorithm~{\EGUP} is a $3$-approximation algorithm.
\end{theorem}

\begin{proof}
  By Property~(SI.\ref{SI:primary disjoint}), different
  demands from the same client are assigned to different
  primary demands, and by (PD.\ref{PD:disjoint}) each primary
  demand opens a different facility. This ensures that our
  solution is feasible, namely each client $j$ is connected
  to $r_j$ different facilities (some possibly located on
  the same site).  As for the total cost,
  Lemma~\ref{lemma:3fac} and Lemma~\ref{lemma:3dist} imply
  that the total cost is at most
  $F^\ast+C^\ast+2\cdot\LP^\ast = 3\cdot\LP^\ast \leq
  3\cdot\OPT$.
\end{proof}



\section{Algorithm~{\ECHS} with Ratio $1.736$}\label{sec: 1.736-approximation}

In this section we improve the approximation ratio to $1+2/e \approx
1.736$. The improvement comes from a slightly modified rounding
process and refined analysis.  Note that the facility opening cost of
Algorithm~{\EGUP} does not exceed that of the fractional optimum
solution, while the connection cost could be far from the optimum,
since we connect a non-primary demand to a facility in the neighborhood of
its assigned primary demand and then estimate the distance using the
triangle inequality. The basic idea to improve the estimate of the connection cost,
following the approach of Chudak and Shmoys~\cite{ChudakS04}, 
is to connect each non-primary demand to its
nearest neighbor when one is available, and to only use the facility opened by
its assigned primary demand when none of its neighbors is open.


\paragraph{Algorithm~{\ECHS}.}
As before,
the algorithm starts by solving the linear program and applying the
adaptive partitioning algorithm  described in 
Section~\ref{sec: adaptive partitioning} to obtain a partitioned
solution $(\barbfx, \barbfy)$. Then we apply the rounding
process to compute an integral solution (see Pseudocode~\ref{alg:lpr3}).  

We start, as before, by opening exactly one facility $\phi(\kappa)$ in the 
neighborhood of each primary demand $\kappa$ (Line 2).  For any
non-primary demand $\nu$ assigned to $\kappa$, we refer to
$\phi(\kappa)$ as the \emph{target} facility of $\nu$.  In
Algorithm~{\EGUP}, $\nu$ was connected to $\phi(\kappa)$,
but in Algorithm~{\ECHS} we may be able to find an open
facility in $\nu$'s neighborhood and connect $\nu$ to this
facility.  Specifically, the two changes in the
algorithm are as follows:
\begin{description}
\item{(1)} Each facility $\mu$ that is not in the neighborhood of any
  primary demand is opened, independently, with probability
  $\bary_{\mu}$ (Lines 4--5). Notice that if $\bary_\mu>0$ then, due
  to completeness of the partitioned fractional solution, we have
  $\bary_{\mu}= \barx_{\mu\nu}$ for some demand $\nu$. This implies
  that $\bary_{\mu}\leq 1$, because $\barx_{\mu\nu}\le 1$, by
  (PS.\ref{PS:one}).
\item{(2)} When connecting demands to facilities, a primary demand
  $\kappa$ is connected to the only facility $\phi(\kappa)$ opened in
  its neighborhood, as before (Line 3).  For a non-primary demand
  $\nu$, if its neighborhood $\wbarN(\nu)$ has an open facility, we
  connect $\nu$ to the closest open facility in $\wbarN(\nu)$ (Line
  8). Otherwise, we connect $\nu$ to its target facility (Line 10).
\end{description}


\begin{algorithm}
  \caption{Algorithm~{\ECHS}:
    Constructing Integral Solution}
  \label{alg:lpr3}
  \begin{algorithmic}[1]
    \For{each $\kappa\in P$} 
    \State choose one $\phi(\kappa)\in \wbarN(\kappa)$,
    with each $\mu\in\wbarN(\kappa)$ chosen as $\phi(\kappa)$
    with probability $\bary_\mu$ 
    \State open $\phi(\kappa)$ and connect $\kappa$ to $\phi(\kappa)$
    \EndFor
    \For{each $\mu\in\facilityset - \bigcup_{\kappa\in P}\wbarN(\kappa)$} 
    \State open $\mu$ with probability $\bary_\mu$ (independently)
    \EndFor
    \For{each non-primary demand $\nu\in\demandset$}
    \If{any facility in $\wbarN(\nu)$ is open}
    \State{connect $\nu$ to the nearest open facility in $\wbarN(\nu)$}
    \Else
    \State connect $\nu$ to $\phi(\kappa)$ where $\kappa$ is $\nu$'s
     assigned primary demand
    \EndIf
    \EndFor
  \end{algorithmic}
\end{algorithm}


\paragraph{Analysis.}
We shall first argue that the integral solution thus
constructed is feasible, and then we bound the total cost of
the solution. Regarding feasibility, the only constraint
that is not explicitly enforced by the algorithm is the
fault-tolerance requirement; namely that each client $j$ is
connected to $r_j$ different facilities. Let $\nu$ and
$\nu'$ be two different sibling demands of client $j$ and let
their assigned primary demands be $\kappa$ and $\kappa'$
respectively. Due to (SI.\ref{SI:primary
  disjoint}) we know $\kappa \neq \kappa'$. From
(SI.\ref{SI:siblings disjoint}) we have $\wbarN(\nu) \cap
\wbarN(\nu') = \emptyset$. From (SI.\ref{SI:primary
  disjoint}), we have $\wbarN(\nu) \cap \wbarN(\kappa') =
\emptyset$ and $\wbarN(\nu') \cap \wbarN(\kappa) =
\emptyset$. From (PD.\ref{PD:disjoint}) we have
$\wbarN(\kappa)\cap \wbarN(\kappa') = \emptyset$. It follows
that $(\wbarN(\nu) \cup \wbarN(\kappa)) \cap (\wbarN(\nu')
\cup \wbarN(\kappa')) = \emptyset$. Since the algorithm
connects $\nu$ to some facility in $\wbarN(\nu) \cup
\wbarN(\kappa)$ and $\nu'$ to some facility in $\wbarN(\nu')
\cup \wbarN(\kappa')$, $\nu$ and $\nu'$ will be connected to
different facilities.


\smallskip
We now show that the expected cost of the computed solution is bounded by
$(1+2/e) \cdot \LP^\ast$. By
(PD.\ref{PD:disjoint}), every facility may appear in at
most one primary demand's neighborhood, and the facilities
open in Line~4--5 of Pseudocode~\ref{alg:lpr3} do not appear
in any primary demand's neighborhood. Therefore, by
linearity of expectation, the expected facility cost of
Algorithm~{\ECHS} is 
\begin{equation*}
\Exp[F_{\smallECHS}] 
	= \sum_{\mu\in\facilityset} f_\mu \bary_{\mu} 
	= \sum_{i\in\sitesset} f_i\sum_{\mu\in i} \bary_{\mu} 
	= \sum_{i\in\sitesset} f_i y_i^\ast = F^\ast,
\end{equation*}
where the third equality follows from (PS.\ref{PS:yi}).

\smallskip

To bound the connection cost, we adapt an argument of Chudak
and Shmoys~\cite{ChudakS04}. Consider a demand $\nu$ and denote by $C_\nu$ the
random variable representing the connection cost for $\nu$.
Our goal now is to estimate $\Exp[C_\nu]$, the expected value of $C_\nu$.
Demand $\nu$ can either get connected directly to some facility in
$\wbarN(\nu)$ or indirectly to its target facility $\phi(\kappa)\in
\wbarN(\kappa)$, where $\kappa$ is the primary demand to
which $\nu$ is assigned. We will analyze these two cases separately.

In our analysis, in this section and the next one, we will use notation
\begin{equation*}
D(A,\sigma) {=} \sum_{\mu\in A}
d_{\mu\sigma}\bary_{\mu}/\sum_{\mu\in A} \bary_{\mu}
\end{equation*}
for the average distance between a demand $\sigma$ and a set $A$ of facilities.
Note that, in particular, we have $\concost_\nu = D(\wbarN(\nu),\nu)$.

We first estimate the expected cost $d_{\phi(\kappa)\nu}$ of the indirect
connection. Let $\Lambda^\nu$ denote the event that some 
facility in $\wbarN(\nu)$ is opened. Then
\begin{equation}
	\Exp[C_\nu \mid\neg\Lambda^\nu] 
	=   \Exp[ d_{\phi(\kappa)\nu} \mid \neg\Lambda^\nu] 
	= 	D(\wbarN(\kappa) \setminus \wbarN(\nu), \nu).
			\label{eqn: expected indirect connection}
\end{equation}
Note that $\neg\Lambda^\nu$ implies that $\wbarN(\kappa) \setminus
\wbarN(\nu)\neq\emptyset$, since $\wbarN(\kappa)$ contains
exactly one open facility, namely $\phi(\kappa)$.


\begin{lemma}
  \label{lem:echu indirect}
  Let $\nu$ be a demand assigned to a primary demand $\kappa$, and
assume that $\wbarN(\kappa) \setminus \wbarN(\nu)\neq\emptyset$.
Then 
\begin{equation*}
	\Exp[ C_\nu \mid\neg\Lambda^\nu]  \leq
  		\concost_\nu+2\alpha_{\nu}^\ast.
\end{equation*}
\end{lemma}

\begin{proof}
By (\ref{eqn: expected indirect connection}), we need to show that $D(\wbarN(\kappa)
  \setminus \wbarN(\nu), \nu) \leq \concost_\nu +
  2\alpha_{\nu}^\ast$. There are two cases to consider.

\begin{description}
\item{\mycase{1}}
	 There exists some $\mu'\in \wbarN(\kappa) \cap
  \wbarN(\nu)$ such that $d_{\mu' \kappa} \leq \concost_\kappa$.
In this case, for every $\mu\in \wbarN(\kappa)\setminus \wbarN(\nu)$, we have
\begin{equation*}
d_{\mu \nu} \leq d_{\mu \kappa} + d_{\mu' \kappa} + d_{\mu' \nu}  
 	\le  \alpha^\ast_\kappa + \concost_\kappa + \alpha^\ast_{\nu}
  \leq \concost_\nu + 2\alpha_{\nu}^\ast,
\end{equation*}
using the triangle inequality, complementary slackness, and (PD.\ref{PD:assign:cost}).
By summing over all $\mu\in \wbarN(\kappa) \setminus \wbarN(\nu)$, it
follows that $D(\wbarN(\kappa) \setminus \wbarN(\nu), \nu) \leq
\concost_\nu + 2\alpha_{\nu}^\ast$.

\item{\mycase{2}}
 Every $\mu'\in \wbarN(\kappa)\cap \wbarN(\nu)$
has $d_{\mu'\kappa} > \concost_\kappa$. Since $\concost_{\kappa} = D(\wbarN(\kappa),\kappa)$,
this implies that
$D(\wbarN(\kappa) \setminus \wbarN(\nu),\kappa)\leq \concost_{\kappa}$. Therefore,
choosing an arbitrary $\mu'\in \wbarN(\kappa)\cap \wbarN(\nu)$,
we obtain
\begin{equation*}
  D(\wbarN(\kappa) \setminus \wbarN(\nu), \nu) 
	\leq  D(\wbarN(\kappa) \setminus \wbarN(\nu), \kappa) 
			+ d_{\mu' \kappa} + d_{\mu' \nu} 
	\leq  \concost_{\kappa} +
  \alpha_{\kappa}^\ast + \alpha_{\nu}^\ast
	\leq \concost_\nu + 2\alpha_{\nu}^\ast,
\end{equation*}
where we again use the triangle inequality,
complementary slackness, and  (PD.\ref{PD:assign:cost}).
\end{description}
Since the lemma holds in both cases, the proof is now complete.
\end{proof}

We now continue our estimation of the connection cost.  The next step
of our analysis is to show that 
\begin{equation}
	\Exp[C_\nu]\le \concost_{\nu} + \frac{2}{e}\alpha^\ast_\nu.
	\label{eqn: echs bound for connection cost}
\end{equation}
The argument is divided into three cases. The first, easy case is when
$\nu$ is a primary demand $\kappa$. According to the algorithm
(see Pseudocode~\ref{alg:lpr3}, Line~2), we have $C_\kappa = d_{\mu\kappa}$ with probability $\bary_{\mu}$, 
for $\mu\in \wbarN(\kappa)$. Therefore $\Exp[C_\kappa] = \concost_{\kappa}$, so
(\ref{eqn: echs bound for connection cost}) holds.

Next, we consider a non-primary demand $\nu$. Let $\kappa$
be the primary demand that $\nu$ is assigned to. We first
deal with the sub-case when $\wbarN(\kappa)\setminus
\wbarN(\nu) = \emptyset$, which is the same as
$\wbarN(\kappa) \subseteq \wbarN(\nu)$. Property (CO)
implies that $\barx_{\mu\nu} = \bary_{\mu} =
\barx_{\mu\kappa}$ for every $\mu \in \wbarN(\kappa)$, so we
have $\sum_{\mu\in\wbarN(\kappa)} \barx_{\mu\nu} =
\sum_{\mu\in\wbarN(\kappa)} \barx_{\mu\kappa} = 1$, due to
(PS.\ref{PS:one}). On the other hand, we have
$\sum_{\mu\in\wbarN(\nu)} \barx_{\mu\nu} = 1$, and
$\barx_{\mu\nu} > 0$ for all $\mu\in \wbarN(\nu)$. Therefore
$\wbarN(\kappa) = \wbarN(\nu)$ and $C_\nu$ has exactly the
same distribution as $C_\kappa$.  So this case reduces to
the first case, namely we have $\Exp[C_{\nu}] =
\concost_{\nu}$, and (\ref{eqn: echs bound for connection
  cost}) holds.

The last, and only non-trivial case is when $\wbarN(\kappa)\setminus
\wbarN(\nu)\neq\emptyset$. We handle this case in the following lemma.


\begin{lemma}\label{lem: echs expected C_nu}
Assume that $\wbarN(\kappa) \setminus \wbarN(\nu) \neq \emptyset$.
Then the expected connection cost of $\nu$, conditioned on the event that at least one of 
its neighbor opens, satisfies
\begin{equation*}
  \Exp[C_\nu \mid \Lambda^\nu] \leq \concost_{\nu}.
\end{equation*}
\end{lemma}

\begin{proof}
The proof is similar to an analogous result in~\cite{ChudakS04,ByrkaA10}. 
For the sake of completeness we sketch here a simplified argument, adapted to our
terminology and notation.
The idea is to consider a different random process that is
easier to analyze and whose expected connection cost is not better than that in
the algorithm.

We partition $\wbarN(\nu)$ into groups $G_1,...,G_k$, where two
different facilities $\mu$ and $\mu'$ are put in the same $G_s$, where
$s\in \{1,\ldots,k\}$, if they both belong to the same set
$\wbarN(\kappa)$ for some primary demand $\kappa$. If some $\mu$ is
not a neighbor of any primary demand, then it constitutes a singleton
group.  For each $s$, let $\bard_s = D(G_s,\nu)$ be the average
distance from $\nu$ to $G_s$.  Assume that $G_1,...,G_k$ are ordered
by nondecreasing average distance to $\nu$, that is $\bard_1 \le
\bard_2 \le ... \le \bard_k$.  For each group $G_s$, we select it,
independently, with probability $g_s = \sum_{\mu\in G_s}\bary_{\mu}$.
For each selected group $G_s$,  we
open exactly one facility in $G_s$, where each $\mu\in G_s$
is opened with probability $\bary_{\mu}/\sum_{\eta\in G_s}
\bary_{\eta}$.

So far, this process is the same as that in the algorithm (if restricted to $\wbarN(\nu)$).
However, we connect $\nu$ in a slightly different way, by choosing the smallest
$s$ for which $G_s$ was selected and connecting $\nu$ to the open facility in $G_s$.
This can only increase our expected connection cost, assuming that at least one
facility in $\wbarN(\nu)$ opens, so
\begin{align}
  \Exp[C_\nu \mid \Lambda^\nu] &\leq \frac{1}{\Prob[\Lambda^\nu]}
  \left( \bard_1 g_1 + \bard_2 g_2 (1-g_1) + \ldots + \bard_k g_k
    (1-g_1) (1-g_2) \ldots (1-g_k) \right)
			\notag
  \\
  &\leq \frac{1}{\Prob[\Lambda^\nu]}
	\cdot \sum_{s=1}^k \bard_s g_s
	\cdot
		\left(\sum_{t=1}^k g_t \prod_{z=1}^{t-1} (1-g_z)\right)
			\label{eqn: echs ineq direct cost, step 1}
  \\
  &= \sum_{s=1}^k \bard_s g_s
			\label{eqn: echs ineq direct cost, step 2}
	\\
			&= \concost_{\nu}.
				\label{eqn: echs ineq direct cost, step 3}
\end{align}
The proof for inequality (\ref{eqn: echs ineq direct cost, step 1}) 
is given in \ref{sec: ECHSinequality} (note that $\sum_{s=1}^k g_s = 1$),
equality (\ref{eqn: echs ineq direct cost, step 2}) follows from
$\Prob[\Lambda^\nu] = 1 - \prod_{t=1}^k (1-g_t)
					= \sum_{t=1}^k g_t
                                        \prod_{z=1}^{t-1} (1 - g_z)$,
and (\ref{eqn: echs ineq direct cost, step 3}) follows from the definition
of the distances $\bard_s$, probabilities $g_s$, and simple algebra.
\end{proof}

Next, we show an estimate on the probability that none of $\nu$'s
neighbors is opened by the algorithm.

\begin{lemma}\label{lem: probability of not Lambda^nu}
The probability that none of $\nu$'s neighbors is opened satisfies
$\Prob[\neg\Lambda^\nu] \le 1/e$.
\end{lemma}

\begin{proof}
We use the same partition of $\wbarN(\nu)$ into groups $G_1,...,G_k$ as
in the proof of Lemma~\ref{lem: echs expected C_nu}. Denoting by
$g_s$ the probability that a group $G_s$ is selected (and thus that it
has an open facility), we have
\begin{equation*}
\Prob[\neg\Lambda^\nu] = \prod_{s=1}^k (1 - g_s)
			\le e^{- \sum_{s=1}^k g_s}
			= e^{-\sum_{\mu \in \wbarN(\nu)} \bary_{\mu}}
			= \frac{1}{e}.
\end{equation*}
In this derivation, we first use that $1-x\le e^{-x}$ holds for all $x$,
the second equality follows from $\sum_{s=1}^k g_s = \sum_{\mu \in \wbarN(\nu)} \bary_{\mu}$
and the last equality follows from 
$\sum_{\mu \in \wbarN(\nu)} \bary_{\mu} = 1$.
\end{proof}

We are now ready to estimate the unconditional expected connection cost of $\nu$
(in the case when $\wbarN(\kappa)\setminus \wbarN(\nu)\neq\emptyset$)
as follows:
\begin{align}
  \notag
  \Exp[C_\nu] &= \Exp[C_{\nu} \mid \Lambda^\nu] \cdot \Prob[\Lambda^\nu] 
	+ \Exp[C_{\nu} \mid \neg \Lambda^\nu] \cdot	\Prob[\neg \Lambda^\nu]
  \\
  &\leq \concost_{\nu} \cdot \Prob[\Lambda^\nu] 
		+ (\concost_{\nu} + 2\alpha_{\nu}^\ast)  \cdot \Prob[\neg \Lambda^\nu]
  \label{eqn: Cnu estimate 0}
  \\
  &= \concost_{\nu} 
	+  2\alpha_{\nu}^\ast \cdot \Prob[\neg \Lambda^\nu]
		\notag
	\\
	&\le \concost_{\nu} + \frac{2}{e}\cdot\alpha_{\nu}^\ast.
	  \label{eqn: Cnu estimate last}
\end{align}
In the above derivation, inequality (\ref{eqn: Cnu estimate 0})
follows from Lemmas~\ref{lem:echu indirect} and \ref{lem: echs expected C_nu}, 
and inequality (\ref{eqn: Cnu estimate last}) follows from
Lemma~\ref{lem: probability of not Lambda^nu}.

\medskip

We have thus shown that the bound (\ref{eqn: echs bound for connection cost})
holds in all three cases.
Summing over all demands $\nu$ of a client $j$, we can now bound
the expected connection cost of client $j$:
\begin{equation*}
  \Exp[C_j] = \textstyle\sum_{\nu\in j} \Exp[C_\nu] 
\leq {\textstyle\sum_{\nu\in j} (\concost_{\nu} + \frac{2}{e}\cdot\alpha_{\nu}^\ast) }
  = { C_j^\ast + \frac{2}{e}\cdot r_j\alpha_j^\ast}.
\end{equation*}
Finally, summing over all clients $j$, we obtain our bound on
the expected connection cost,
\begin{equation*}
 \Exp[ C_{\smallECHS}] \le C^\ast + \frac{2}{e}\cdot\LP^\ast.
\end{equation*}
Therefore we have established that
our algorithm constructs a feasible integral solution with
an overall expected cost 
\begin{equation*}
  \label{eq:chudakall}
	 \Exp[ F_{\smallECHS} + C_{\smallECHS}]
	\le
  	F^\ast + C^\ast + \frac{2}{e}\cdot \LP^\ast = (1+2/e)\cdot \LP^\ast
  \leq (1+2/e)\cdot \OPT.
\end{equation*}
Summarizing, we obtain the main result of this section.

\begin{theorem}\label{thm:1736}
  Algorithm~{\ECHS} is a $(1+2/e)$-approximation algorithm for \FTFP.
\end{theorem}


\section{Algorithm~{\EBGS} with Ratio $1.575$}\label{sec: 1.575-approximation}

In this section we give our main result, a $1.575$-approximation
algorithm for $\FTFP$, where $1.575$ is the value of $\min_{\gamma\geq
  1}\max\{\gamma, 1+2/e^\gamma, \frac{1/e+1/e^\gamma}{1-1/\gamma}\}$,
rounded to three decimal digits. This matches the ratio of the best
known LP-rounding algorithm for UFL by
Byrka~{\etal}~\cite{ByrkaGS10}. 

Recall that in Section~\ref{sec: 1.736-approximation} we showed how to
compute an integral solution with facility cost bounded by $F^\ast$
and connection cost bounded by $C^\ast + 2/e\cdot\LP^\ast$. Thus,
while our facility cost does not exceed the optimal fractional
facility cost, our connection cost is significantly larger than the
connection cost in the optimal fractional solution.  A natural idea is
to balance these two ratios by reducing the connection cost at the
expense of the facility cost. One way to do this would be to increase
the probability of opening facilities, from $\bary_{\mu}$ (used in
Algorithm~{\ECHS}) to, say, $\gamma\bary_{\mu}$, for some $\gamma >
1$. This increases the expected facility cost by a factor of $\gamma$
but, as it turns out, it also reduces the probability that an indirect
connection occurs for a non-primary demand to $1/e^\gamma$ (from the
previous value $1/e$ in {\ECHS}). As a consequence, for each primary
demand $\kappa$, the new algorithm will select a facility to open from
the nearest facilities $\mu$ in $\wbarN(\kappa)$ such that the
connection values $\barx_{\mu\nu}$ sum up to $1/\gamma$, instead of
$1$ as in Algorithm {\ECHS}. It is easily seen that this will improve
the estimate on connection cost for primary demands.  These two
changes, along with a more refined analysis, are the essence of the
approach in~\cite{ByrkaGS10}, expressed in our terminology.

Our approach can be thought of as a combination of the above ideas
with the techniques of demand reduction and
adaptive partitioning that we introduced earlier. However, our
adaptive partitioning technique needs to be carefully modified,
because now we will be using a more intricate neighborhood structure,
with the neighborhood of each demand divided into two disjoint parts,
and with restrictions on how parts from different demands can overlap.

We begin by describing properties that our partitioned fractional
solution $(\barbfx,\barbfy)$ needs to satisfy. Assume that $\gamma$ is
some constant such that $1 < \gamma < 2$. As mentioned earlier,
the neighborhood $\wbarN(\nu)$ of each demand $\nu$ will be divided
into two disjoint parts.  The first part, called the \emph{close
  neighborhood} and denoted $\wbarclsnb(\nu)$, contains the facilities
in $\wbarN(\nu)$ nearest to $\nu$ with the total connection value
equal $1/\gamma$, that is $\sum_{\mu\in\wbarclsnb(\nu)} \barx_{\mu\nu}
= 1/\gamma$.  The second part, called the \emph{far neighborhood} and
denoted $\wbarfarnb(\nu)$, contains the remaining facilities in
$\wbarN(\nu)$ (so $\sum_{\mu\in\wbarfarnb(\nu)} \barx_{\mu\nu} = 1-1/\gamma$).  We
restate these definitions formally below in Property~(NB).  Recall
that for any set $A$ of facilities and a demand $\nu$, by
$D(A,\nu)$ we denote the average distance between $\nu$ and the
facilities in $A$, that is $D(A,\nu) =\sum_{\mu\in A}
d_{\mu\nu}\bary_{\mu}/\sum_{\mu\in A} \bary_{\mu}$.  We will use
notations $\clsdist(\nu)=D(\wbarclsnb(\nu),\nu)$ and
$\fardist(\nu)=D(\wbarfarnb(\nu),\nu)$ for the average distances from
$\nu$ to its close and far neighborhoods, respectively.  By the
definition of these sets and the completeness property (CO), these
distances can be expressed as
\begin{equation*}
\clsdist(\nu)=\gamma\sum_{\mu\in\wbarclsnb(\nu)}
			d_{\mu\nu}\barx_{\mu\nu} \quad\text{and}\quad
\fardist(\nu)=\frac{\gamma}{\gamma-1}\sum_{\mu\in\wbarfarnb(\nu)}
d_{\mu\nu}\barx_{\mu\nu}. 
\end{equation*}
We will also use notation $\clsmax(\nu)=\max_{\mu\in\wbarclsnb(\nu)}
d_{\mu\nu}$ for the maximum distance from $\nu$ to its close
neighborhood. The average distance from a demand $\nu$ to its overall
neighborhood $\wbarN(\nu)$ is denoted as $\concost(\nu) =
D(\wbarN(\nu), \nu) = \sum_{\mu \in \wbarN(\nu)} d_{\mu\nu}
\barx_{\mu\nu}$. It is easy to see that
\begin{equation}
  \concost(\nu) = \frac{1}{\gamma} \clsdist(\nu) + \frac{\gamma -
    1}{\gamma} \fardist(\nu).
  \label{eqn:avg dist cls dist far dist}
\end{equation}

Our partitioned solution $(\barbfx,\barbfy)$ must satisfy the same
partitioning and completeness properties as before, namely properties
(PS) and (CO) in Section~\ref{sec: adaptive partitioning}.  In
addition, it must satisfy a new neighborhood property (NB) and modified
properties (PD') and (SI'), listed below.

\begin{description}
	
      \renewcommand{\theenumii}{(\alph{enumii})}
      \renewcommand{\labelenumii}{\theenumii}

\item{(NB)} \label{NB}
	\emph{Neighborhoods.}
	For each demand $\nu \in \demandset$, its neighborhood is divided into \emph{close} and
	\emph{far} neighborhood, that is $\wbarN(\nu) = \wbarclsnb(\nu) \cup \wbarfarnb(\nu)$, where
	\begin{itemize}
	\item $\wbarclsnb(\nu) \cap \wbarfarnb(\nu) = \emptyset$,
	\item $\sum_{\mu\in\wbarclsnb(\nu)} \barx_{\mu\nu} =1/\gamma$, and 
	\item if $\mu\in \wbarclsnb(\nu)$ and $\mu'\in \wbarfarnb(\nu)$ 
				then $d_{\mu\nu}\le d_{\mu'\nu}$.   
	\end{itemize}
	Note that the first two conditions, together with
        (PS.\ref{PS:one}), imply that $\sum_{\mu\in\wbarfarnb(\nu)}
        \barx_{\mu\nu} =1-1/\gamma$. When defining $\wbarclsnb(\nu)$,
        in case of ties, which can occur when some facilities in
        $\wbarN(\nu)$ are at the same distance from $\nu$, we use a
        tie-breaking rule that is explained in the proof of
        Lemma~\ref{lem: PD1: primary overlap} (the only place where
        the rule is needed).

\item{(PD')} \emph{Primary demands.}
	Primary demands satisfy the following conditions:

	\begin{enumerate}
		
	\item\label{PD1:disjoint}  For any two different primary demands $\kappa,\kappa'\in P$ we have
				$\wbarclsnb(\kappa)\cap \wbarclsnb(\kappa') = \emptyset$.

	\item \label{PD1:yi} For each site $i\in\sitesset$, 
		$ \sum_{\kappa\in P}\sum_{\mu\in
                  i\cap\wbarclsnb(\kappa)}\barx_{\mu\kappa} \leq
                y_i^\ast$. In the summation, as before, we overload notation $i$ to stand for the set of
						facilities created on site $i$.
		
	\item \label{PD1:assign} Each demand $\nu\in\demandset$ is assigned
        to one primary demand $\kappa\in P$ such that

  			\begin{enumerate}
	
				\item \label{PD1:assign:overlap} $\wbarclsnb(\nu) \cap \wbarclsnb(\kappa) \neq \emptyset$, and
				\item \label{PD1:assign:cost}
          $\clsdist(\nu)+\clsmax(\nu) \geq
          \clsdist(\kappa)+\clsmax(\kappa)$.
			\end{enumerate}

	\end{enumerate}
	
\item{(SI')} \emph{Siblings}. For any pair $\nu,\nu'\in\demandset$ of different siblings we have
  \begin{enumerate}

	\item \label{SI1:siblings disjoint}
		  $\wbarN(\nu)\cap \wbarN(\nu') = \emptyset$.
		
	\item \label{SI1:primary disjoint} If $\nu$ is assigned to a primary demand $\kappa$ then
 		$\wbarN(\nu')\cap \wbarclsnb(\kappa) = \emptyset$. In particular, by Property~(PD'.\ref{PD1:assign:overlap}),
		this implies that different sibling demands are assigned to different primary demands, since $\wbarclsnb(\nu')$ is a subset of $\wbarN(\nu')$.

	\end{enumerate}
	
\end{description}


\paragraph{Modified adaptive partitioning.}
To obtain a fractional solution with the above properties, we employ a
modified adaptive partitioning algorithm. As in Section~\ref{sec:
  adaptive partitioning}, we have two phases.  In Phase~1 we split
clients into demands and create facilities on sites, while in Phase~2
we augment each demand's connection values $\barx_{\mu\nu}$ so that the total connection
value of each demand $\nu$ is $1$. As the partitioning algorithm proceeds, for any demand $\nu$,
$\wbarN(\nu)$ denotes the set of facilities with $\barx_{\mu\nu} > 0$;
hence the notation $\wbarN(\nu)$ actually represents a dynamic set which gets fixed 
once the partitioning algorithm concludes both Phase 2. On the
other hand, $\wbarclsnb(\nu)$ and $\wbarfarnb(\nu)$ refer to the close
and far neighborhoods at the time when $\wbarN(\nu)$ is fixed.

Similar to the algorithm in Section~\ref{sec: adaptive partitioning},
Phase~1 runs in iterations. Fix some iteration and consider any client
$j$.  As before, $\wtildeN(j)$ is the neighborhood of $j$ with respect
to the yet unpartitioned solution, namely the set of facilities $\mu$
such that $\tildex_{\mu j}>0$. Order the facilities in this set as
$\wtildeN(j) = \braced{\mu_1,...,\mu_q}$ with non-decreasing distance
from $j$, that is $d_{\mu_1 j} \leq d_{\mu_2 j} \leq \ldots \leq
d_{\mu_q j}$. Without loss of generality,
there is an index $l$ for which $\sum_{s=1}^l \tildex_{\mu_s j} =
1/\gamma$, since we can always split one facility to achieve
this. Then we define $\wtildeclsnb(j) = \braced{\mu_1,...,\mu_l}$. 
(Unlike close neighborhoods of demands, $\wtildeclsnb(j)$ can vary over time.)
We also use notation
\begin{equation*}
\tcccls(j) =  D(\wtildeclsnb(j), j) = \gamma\sum_{\mu\in\wtildeclsnb(j)} d_{\mu j} \tildex_{\mu j}
			\quad\textrm{ and }\quad
 \dmaxcls(j) = \max_{\mu \in \wtildeclsnb(j)} d_{\mu j}. 
\end{equation*}

When the iteration starts, we first find a not-yet-exhausted client
$p$ that minimizes the value of $\tcccls(p) + \dmaxcls(p)$ and create
a new demand $\nu$ for $p$.  Now we have two cases:
\begin{description}
\item{\mycase{1}} $\wtildeclsnb(p) \cap \wbarN(\kappa)\neq\emptyset$
  for some existing primary demand $\kappa\in P$.  In this case we
  assign $\nu$ to $\kappa$. As before, if there are multiple such
  $\kappa$, we pick any of them. We also fix $\barx_{\mu \nu} \assign
  \tildex_{\mu p}$ and $\tildex_{\mu p}\assign 0$ for each $\mu \in
  \wtildeN(p)\cap \wbarN(\kappa)$. Note that although we
  check for overlap between $\wtildeclsnb(p)$ and $\wbarN(\kappa)$,
  the facilities we actually move into $\wbarN(\nu)$ include all
  facilities in the intersection of $\wtildeN(p)$, a bigger set, with
  $\wbarN(\kappa)$.

  At this time, the total connection value 
	between $\nu$ and $\mu\in \wbarN(\nu)$ is at most $1/\gamma$,
	 since $\sum_{\mu \in \wbarN(\kappa)}\bary_{\mu} = 1/\gamma$ 
	(this follows from the definition of neighborhoods for new primary demands in Case~2 below) 
	and  we have $\wbarN(\nu) \subseteq \wbarN(\kappa)$ at this point. Later
  in Phase 2 we will add additional facilities from $\wtildeN(p)$ to
  $\wbarN(\nu)$ to make $\nu$'s total connection value equal to $1$.

\item{\mycase{2}} $\wtildeclsnb(p) \cap \wbarN(\kappa) = \emptyset$
  for all existing primary demands $\kappa\in P$.  In this case we
  make $\nu$ a primary demand (that is, add it to $P$) and assign it
  to itself.  We then move the facilities from $\wtildeclsnb(p)$ to
  $\wbarN(\nu)$, that is for $\mu \in \wtildeclsnb(p)$ we set
  $\barx_{\mu \nu}\assign \tildex_{\mu p}$ and $\tildex_{\mu p}\set
  0$.

  It is easy to see that the total connection value of $\nu$ to
  $\wbarN(\nu)$ is now exactly $1/\gamma$, that is
	$\sum_{\mu \in \wbarN(\nu)}\bary_{\mu} = 1/\gamma$.
Moreover, facilities
  remaining in $\wtildeN(p)$ are all farther away from $\nu$ than
  those in $\wbarN(\nu)$. As we add only facilities from $\wtildeN(p)$
  to $\wbarN(\nu)$ in Phase~2, the final $\wbarclsnb(\nu)$ contains
  the same set of facilities as the current set $\wbarN(\nu)$.
  (More precisely, $\wbarclsnb(\nu)$ consists of the facilities that
	either are currently in $\wbarN(\nu)$ or were obtained from splitting
	the facilities currently in $\wbarN(\nu)$.)
\end{description}
Once all clients are exhausted, that is, each client $j$ has $r_j$
demands created, Phase~1 concludes. We then run Phase~2, the
augmenting phase, following the same steps as in Section~\ref{sec:
  adaptive partitioning}.  For each client $j$ and each demand $\nu\in
j$ with total connection value to $\wbarN(\nu)$ less than $1$
(that is, $\sum_{\mu\in\wbarN(\nu)} \barx_{\mu\nu} < 1$),
we use our $\AugmentToUnit()$
procedure to add additional facilities (possibly split, if necessary)
from $\wtildeN(j)$ to $\wbarN(\nu)$ to make the total connection value
between $\nu$ and $\wbarN(\nu)$ equal $1$.

\medskip

This completes the description of the partitioning
algorithm. Summarizing, for each client $j\in\clientset$ we 
created $r_j$ demands on the same point as $j$, and we created a number
of facilities at each site $i\in\sitesset$. Thus computed sets of
demands and facilities are denoted $\demandset$ and $\facilityset$,
respectively.  For each facility $\mu\in i$ we defined its fractional
opening value $\bary_\mu$, $0\le \bary_\mu\le 1$, and for each demand
$\nu\in j$ we defined its fractional connection value
$\barx_{\mu\nu}\in \braced{0,\bary_\mu}$.  The connections with
$\barx_{\mu\nu} > 0$ define the neighborhood $\wbarN(\nu)$. The facilities in
$\wbarN(\nu)$ that are closest to $\nu$ and have total connection value from $\nu$ equal
$1/\gamma$ form the close neighborhood $\wbarclsnb(\nu)$, while the remaining facilities
in $\wbarN(\nu)$ form the far neighborhood
$\wbarfarnb(\nu)$. It remains to show that this partitioning satisfies all the desired
properties.


\medskip
\paragraph{Correctness of partitioning.}
We now argue that our partitioned fractional solution $(\barbfx,\barbfy)$
satisfies all the stated properties. Properties~(PS), (CO) and (NB) are
directly enforced by the algorithm.

(PD'.\ref{PD1:disjoint}) holds because for each primary demand
$\kappa\in p$, $\wbarclsnb(\kappa)$ is the same set as
$\wtildeclsnb(p)$ at the time when $\kappa$ was created, and
$\wtildeclsnb(p)$ is removed from $\wtildeN(p)$ right after this
step. Further, the partitioning algorithm makes $\kappa$ a primary
demand only if $\wtildeclsnb(p)$ is disjoint from the set
$\wbarN(\kappa')$ of all existing primary demands $\kappa'$ at that
iteration, but these neighborhoods are the same as the final close
neighborhoods $\wbarclsnb(\kappa')$.

The justification of (PD'.\ref{PD1:yi}) is similar to that for
(PD.\ref{PD:yi}) from Section~\ref{sec: adaptive partitioning}. All
close neighborhoods of primary demands are disjoint, due to
(PD'.\ref{PD1:disjoint}), so each facility $\mu \in i$ can appear in
at most one $\wbarclsnb(\kappa)$, for some $\kappa\in P$. Condition
(CO) implies that $\bary_{\mu} = \barx_{\mu\kappa}$ for $\mu \in \wbarclsnb(\kappa)$.
As a result, the summation on
the left-hand side is not larger than $\sum_{\mu\in i}\bary_{\mu} = y_i^\ast$.

Regarding (PD'.\ref{PD1:assign:overlap}), at first glance this
property seems to follow directly from the algorithm, as we only
assign a demand $\nu$ to a primary demand $\kappa$ when $\wbarN(\nu)$
at that iteration overlaps with $\wbarN(\kappa)$ (which is equal to
the final value of $\wbarclsnb(\kappa)$).  However, it is a little
more subtle, as the final $\wbarclsnb(\nu)$ may contain facilities
added to $\wbarN(\nu)$ in Phase 2. Those facilities may turn out to be
closer to $\nu$ than some facilities in $\wbarN(\kappa) \cap
\wtildeN(j) $ (not $\wtildeN_{\cls}(j)$) that we added to
$\wbarN(\nu)$ in Phase 1. If the final $\wbarclsnb(\nu)$ consists only of
facilities added in Phase 2, we no longer have the desired overlap of
$\wbarclsnb(\kappa)$ and $\wbarclsnb(\nu)$. Luckily this bad scenario
never occurs. We postpone the proof of this property to
Lemma~\ref{lem: PD1: primary overlap}.  The proof of
(PD'.\ref{PD1:assign:cost}) is similar to that of Lemma~\ref{lem:
  PD:assign:cost holds}, and we defer it to Lemma~\ref{lem: PD1:
  primary optimal}.

(SI'.\ref{SI1:siblings disjoint}) follows directly from the algorithm
because for each demand $\nu\in j$, all facilities added to
$\wbarN(\nu)$ are immediately removed from $\wtildeN(j)$ and each
facility is added to $\wbarN(\nu)$ of exactly one demand $\nu \in j$.
Splitting facilities obviously preserves (SI'.\ref{SI1:siblings disjoint}).

The proof of (SI'.\ref{SI1:primary disjoint}) is similar to that of
Lemma~\ref{lem: property SI:primary disjoint holds}. If $\kappa=\nu$
then (SI'.\ref{SI1:primary disjoint}) follows from
(SI'.\ref{SI1:siblings disjoint}), so we can assume that
$\kappa\neq\nu$.  Suppose that $\nu'\in j$ is assigned to $\kappa'\in
P$ and consider the situation after Phase~1. By the way we reassign
facilities in Case~1, at this time we have $\wbarN(\nu)\subseteq
\wbarN(\kappa) = \wbarclsnb(\kappa)$ and $\wbarN(\nu')\subseteq
\wbarN(\kappa') =\wbarclsnb(\kappa')$, so $\wbarN(\nu')\cap
\wbarclsnb(\kappa) = \emptyset$, by (PD'.\ref{PD1:disjoint}).
Moreover, we have $\wtildeN(j) \cap \wbarclsnb(\kappa) = \emptyset$
after this iteration, because any facilities that were also in
$\wbarclsnb(\kappa)$ were removed from $\wtildeN(j)$ when $\nu$ was
created. In Phase~2, augmentation does not change $\wbarclsnb(\kappa)$
and all facilities added to $\wbarN(\nu')$ are from the set
$\wtildeN(j)$ at the end of Phase 1, which is a subset of the set
$\wtildeN(j)$ after this iteration, since $\wtildeN(j)$ can only shrink. 
So the condition (SI'.\ref{SI1:primary disjoint}) will
remain true.


\begin{lemma} \label{lem: PD1: primary overlap}
  Property (PD'.\ref{PD1:assign:overlap}) holds.
\end{lemma}

\begin{proof}
  Let $j$ be the client for which $\nu\in j$. We consider an iteration
  when we create $\nu$ from $j$ and assign it to $\kappa$, and
  within this proof, notation $\wtildeclsnb(j)$ and $\wtildeN(j)$
  will refer to the value of the sets at this particular time.  
At this time, $\wbarN(\nu)$ is initialized to $\wtildeN(j)\cap
  \wbarN(\kappa)$.  Recall that $\wbarN(\kappa)$ is now equal to the
  final $\wbarclsnb(\kappa)$ (taking into account facility splitting). We
  would like to show that the set $\wtildeclsnb(j)\cap
  \wbarclsnb(\kappa)$ (which is not empty) will be included in
  $\wbarclsnb(\nu)$ at the end. Technically speaking, this will not be
  true due to facility splitting, so we need to rephrase this claim
  and the proof in terms of the set of facilities obtained after the
  algorithm completes.

\begin{figure}[ht]
\begin{center}
\includegraphics[width=3.2in]{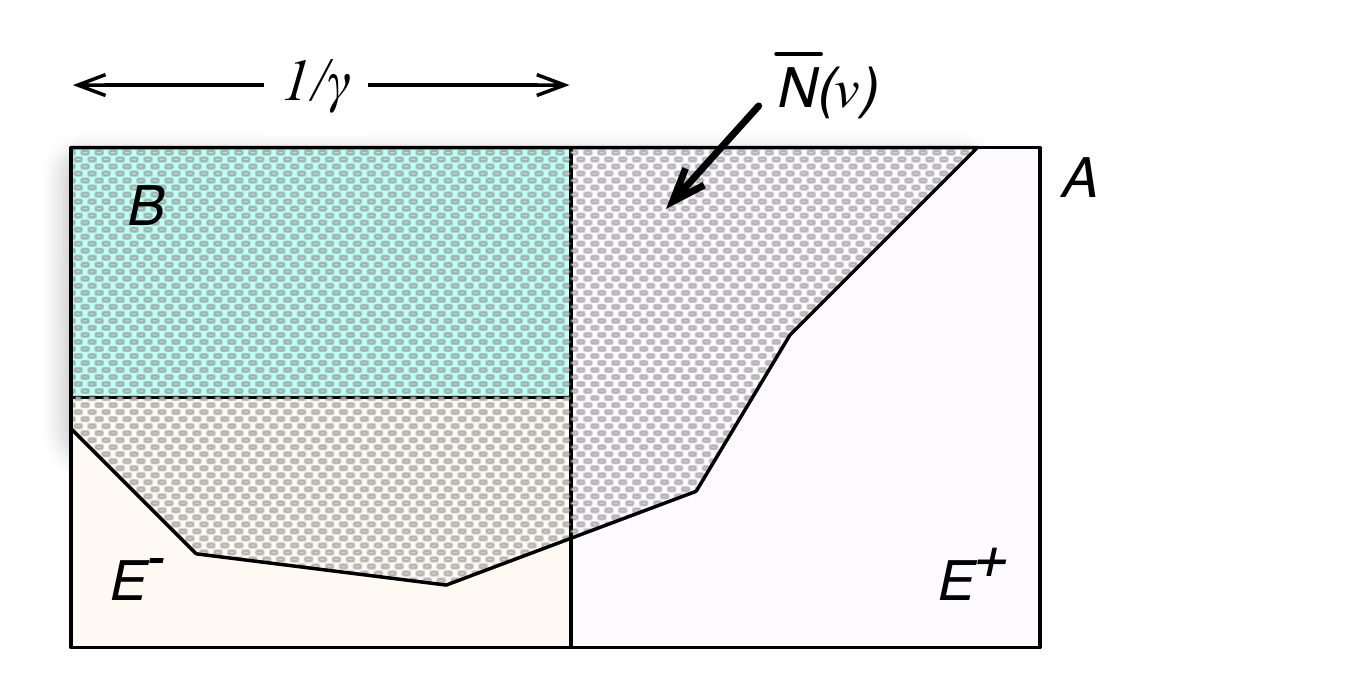}
\caption{Illustration of the sets $\wbarN(\nu)$, $A$, $B$,
  $E^-$ and $E^+$ in the proof of Lemma~\ref{lem: PD1:
    primary overlap}. Let $X \Subset Y$ mean that the facility
	sets $X$ is obtained from $Y$ by splitting facilities.
	We then have $A \Subset \wtildeN(j)$, 
	$B \Subset  \wtildeclsnb(j) \cap \wbarclsnb(\kappa)$, 
	$E^- \Subset  \wtildeclsnb(j) - \wbarclsnb(\kappa)$, 
	$E^+ \Subset \wtildeN(j) - \wtildeclsnb(j)$.}
\label{fig: sets lemma PD'3a}
\end{center}
\end{figure}

  We define the sets $A$, $B$, $E^-$ and $E^+$ as the subsets of
  $\facilityset$ (the final set of facilities) that were obtained from
  splitting facilities in the sets $\wtildeN(j)$, $\wtildeclsnb(j)\cap
  \wbarclsnb(\kappa)$, $\wtildeclsnb(j) - \wbarclsnb(\kappa)$ and
  $\wtildeN(j) - \wtildeclsnb(j)$, respectively.  (See
  Figure~\ref{fig: sets lemma PD'3a}.)  We claim that at the end
  $B\subseteq \wbarclsnb(\nu)$, with the caveat that the ties in the
  definition of $\wbarclsnb(\nu)$ are broken in favor of the
  facilities in $B$.  (This is the tie-breaking rule that we mentioned
  in the definition of $\wbarclsnb(\nu)$.)  This will be sufficient to
  prove the lemma because $B\neq\emptyset$, by the algorithm.

  We now prove this claim. In this paragraph $\wbarN(\nu)$ denotes the
  final set $\wbarN(\nu)$ after both phases are completed. Thus the total
connection value of $\wbarN(\nu)$ to $\nu$ is $1$.
	Note first that
  $B\subseteq \wbarN(\nu) \subseteq A$, because we never remove
  facilities from $\wbarN(\nu)$ and we only add facilities from
  $\wtildeN(j)$.  Also, $B\cup E^-$ represents the facilities obtained
  from $\wtildeclsnb(j)$, so $\sum_{\mu\in B\cup E^-} \bary_{\mu} =
  1/\gamma$.  This and $B\subseteq \wbarN(\nu)$ implies that the total
  connection value of $B\cup (\wbarN(\nu)\cap E^-)$ to $\nu$ is at
  most $1/\gamma$. But all facilities in $B\cup (\wbarN(\nu)\cap E^-)$
  are closer to $\nu$ (taking into account our tie breaking in property (NB))
 	than those in $E^+\cap \wbarN(\nu)$. It follows
  that $B\subseteq \wbarclsnb(\nu)$, completing the proof.
\end{proof}


\begin{lemma}\label{lem: PD1: primary optimal}
  Property (PD'.\ref{PD:assign:cost}) holds.
\end{lemma}

\begin{proof}
This proof is similar to that for Lemma~\ref{lem: PD:assign:cost holds}.
For a client $j$ and demand $\eta$, we will write
$\tcccls^\eta(j)$ and $\dmaxcls^\eta(j)$ to denote the values of
$\tcccls(j)$ and $\dmaxcls(j)$ at the time when $\eta$
was created. (Here $\eta$ may or may not be a demand of client $j$).

Suppose $\nu \in j$ is assigned to a primary demand $\kappa \in p$.
By the way primary demands are constructed in the partitioning
algorithm, $\wtildeclsnb(p)$ becomes $\wbarN(\kappa)$, which is equal
to the final value of $\wbarclsnb(\kappa)$. So we have
$\clsdist(\kappa) = \tcccls^\kappa (p)$ and $\clsmax(\kappa) =
\dmaxcls^\kappa(p)$. Further, since we choose $p$ to minimize
$\tcccls(p) + \dmaxcls(p)$, we have that $\tcccls^\kappa(p) +
\dmaxcls^\kappa(p) \leq \tcccls^\kappa(j) + \dmaxcls^\kappa(j)$.

Using an argument analogous to that in the proof of Lemma~\ref{lem: tcc optimal}, 
our modified partitioning algorithm guarantees that
  $\tcccls^{\kappa}(j) \leq \tcccls^{\nu}(j) \leq \clsdist(\nu)$ and
  $\dmaxcls^{\kappa}(j) \leq \dmaxcls^{\nu}(j) \leq \clsmax(\nu)$ since $\nu$ was
  created later.
  Therefore, we have
  \begin{align*}
    \clsdist(\kappa) + \clsmax(\kappa) &= \tcccls^{\kappa}(p) +	\dmaxcls^{\kappa}(p) 
					\\
					&\leq \tcccls^{\kappa}(j) + \dmaxcls^{\kappa}(j) 
					\leq \tcccls^{\nu}(j) + \dmaxcls^{\nu}(j) 
					\leq \clsdist(\nu) + \clsmax(\nu),
  \end{align*}
completing the proof.
\end{proof}


Now we have completed the proof that the computed partitioning satisfies
all the required properties.

\paragraph{Algorithm~{\EBGS}.}
The complete algorithm starts with solving the LP(\ref{eqn:fac_primal}) and
computing the partitioning described earlier in this section.  Given
the partitioned fractional solution $(\barbfx, \barbfy)$ with the
desired properties, we start the process of opening facilities and
making connections to obtain an integral solution. To this end, for
each primary demand $\kappa\in P$, we open exactly one facility
$\phi(\kappa)$ in $\wbarclsnb(\kappa)$, where each
$\mu\in\wbarclsnb(\kappa)$ is chosen as $\phi(\kappa)$ with
probability $\gamma\bary_{\mu}$. For all facilities
$\mu\in\facilityset - \bigcup_{\kappa\in P}\wbarclsnb(\kappa)$, we
open them independently, each with probability
$\gamma\bary_{\mu}$. 

We claim that all probabilities are well-defined, that is
$\gamma\bary_{\mu} \le 1$ for all $\mu$. Indeed, if $\bary_{\mu}>0$ then
$\bary_{\mu} = \barx_{\mu\nu}$ for some $\nu$, by Property~(CO).
If $\mu\in \wbarclsnb(\nu)$ then the definition of close
neighborhoods implies that $\barx_{\mu\nu} \le 1/\gamma$.
If $\mu\in \wbarfarnb(\nu)$ then
$\barx_{\mu\nu} \le 1-1/\gamma \le 1/\gamma$, because $\gamma < 2$.
Thus $\gamma\bary_{\mu} \le 1$, as claimed.

Next, we connect demands to facilities.  Each primary demand
$\kappa\in P$ will connect to the only open facility $\phi(\kappa)$ in
$\wbarclsnb(\kappa)$.  For each non-primary demand $\nu\in \demandset
- P$, if there is an open facility in $\wbarclsnb(\nu)$ then we
connect $\nu$ to the nearest such facility. Otherwise, we connect
$\nu$ to the nearest far facility in $\wbarfarnb(\nu)$ if one is
open. Otherwise, we connect $\nu$ to its \emph{target facility}
$\phi(\kappa)$, where $\kappa$ is the primary demand that $\nu$ is
assigned to.


\paragraph{Analysis.}
By the algorithm, for each client $j$, all its $r_j$ demands are connected to
open facilities. If two different siblings $\nu,\nu'\in j$ are assigned, respectively,
to primary demands $\kappa$, $\kappa'$ then, by
Properties~(SI'.\ref{SI1:siblings disjoint}), (SI'.\ref{SI1:primary
  disjoint}), and (PD'.\ref{PD1:disjoint}) we have
\begin{equation*}
( \wbarN(\nu) \cup \wbarclsnb(\kappa)) \cap (\wbarN(\nu')\cup \wbarclsnb(\kappa')) = \emptyset.
\end{equation*}
This condition guarantees that $\nu$ and $\nu'$ are assigned to different facilities,
regardless whether they are connected to a neighbor facility or to its target facility.
Therefore the computed solution is feasible.

\medskip

We now estimate the cost of the solution computed by Algorithm {\EBGS}. The lemma
below bounds the expected facility cost.


\begin{lemma} \label{lem: EBGS facility cost}
The expectation of facility cost $F_{\smallEBGS}$ of Algorithm~{\EBGS} is at most $\gamma F^\ast$.
\end{lemma}

\begin{proof}
By the algorithm, each facility $\mu\in \facilityset$ is opened with
probability $\gamma \bary_{\mu}$, independently of whether it belongs to the
close neighborhood of a primary demand or not. Therefore, by
  linearity of expectation, we have that the expected facility cost is
\begin{equation*}
	\Exp[F_{\smallEBGS}] = \sum_{\mu \in \facilityset} f_\mu \gamma \bary_{\mu} 
			= \gamma \sum_{i\in \sitesset} f_i \sum_{\mu\in i} \bary_{\mu} 
			= \gamma \sum_{i \in \sitesset} f_i y_i^\ast = \gamma F^\ast,
\end{equation*}
where the third equality follows from (PS.\ref{PS:yi}).
\end{proof}


\medskip

In the remainder of this section we focus on the connection cost. Let $C_{\nu}$ be the
random variable representing the connection cost of a demand $\nu$. Our objective is
to show that the expectation of $\nu$ satisfies
\begin{equation}
\Exp[C_\nu]	\leq \concost(\nu) \cdot \max\left\{\frac{1/e+1/e^\gamma}{1-1/\gamma}, 1 + \frac{2}{e^\gamma}\right\}.
		\label{eqn: expectation of C_nu for EBGS}
\end{equation}
If $\nu$ is a primary demand then, due to the algorithm, we have $\Exp[C_{\nu}] =
\clsdist(\nu) \le \concost(\nu)$, so (\ref{eqn: expectation of C_nu for EBGS}) is
easily satisfied.

Thus for the rest of the argument we will focus on the case when $\nu$
is a non-primary demand.  Recall that the
algorithm connects $\nu$ to the nearest open facility in
$\wbarclsnb(\nu)$ if at least one facility in $\wbarclsnb(\nu)$ is
open. Otherwise the algorithm connects $\nu$ to the nearest open
facility in $\wbarfarnb(\nu)$, if any. In the event that no facility in
$\wbarN(\nu)$ opens, the algorithm will connect $\nu$ to its target
facility $\phi(\kappa)$, where $\kappa$ is the primary demand that
$\nu$ was assigned to, and $\phi(\kappa)$ is the only facility open in
$\wbarclsnb(\kappa)$. Let $\Lambda^\nu$ denote the event that at least
one facility in $\wbarN(\nu)$ is open and $\Lambda^\nu_{\cls}$ be the
event that at least one facility in $\wbarclsnb(\nu)$ is open.
$\neg \Lambda^\nu$ denotes the complement event of $\Lambda^\nu$, that is,
the event that none of $\nu$'s neighbors opens. 
We want to estimate the following three conditional expectations: 
\begin{equation*}
  \Exp[C_{\nu} \mid
  \Lambda^\nu_{\cls}],\quad \Exp[C_{\nu} \mid \Lambda^\nu \wedge \neg
  \Lambda^\nu_{\cls}], \quad\text{and}\quad \Exp[C_{\nu} \mid \neg \Lambda^\nu], 
\end{equation*}
and their associated probabilities.

We start with a lemma dealing with the third expectation,
$\Exp[C_\nu\mid\neg \Lambda^{\nu}] = \Exp[d_{\phi(\kappa)\nu} \mid
\Lambda^{\nu}]$. The proof of this lemma relies on
Properties~(PD'.\ref{PD1:assign:overlap}) and
(PD'.\ref{PD1:assign:cost}) of modified partitioning and follows the
reasoning in the proof of a similar lemma
in~\cite{ByrkaGS10,ByrkaA10}.  For the sake of completeness, we
include a proof in~\ref{sec: proof of lemma 15}.


\begin{lemma}\label{lem: EBGS target connection cost}
Assuming that no facility in $\wbarN(\nu)$ opens, the expected connection
cost of $\nu$ is
\begin{equation}
  \Exp[C_{\nu} \mid \neg \Lambda^{\nu}] \leq
  \clsdist(\nu) + 2\fardist(\nu).
  \label{eqn: expected connection cost target facility}
\end{equation}
\end{lemma}
\begin{proof}
See \ref{sec: proof of lemma 15}.
\end{proof}

Next, we derive some estimates for the expected cost of direct
connections.  The next technical lemma is a generalization of
Lemma~\ref{lem: echs expected C_nu}. In Lemma~\ref{lem: echs expected
  C_nu} we bound the expected distance to the closest open facility in
$\wbarN(\nu)$, conditioned on at least one facility in $\wbarN(\nu)$
being open. The lemma below provides a similar estimate for an
arbitrary set $A$ of facilities in $\wbarN(\nu)$, conditioned on that
at least one facility in set $A$ is open.  Recall that $D(A,\nu) =
\sum_{\mu \in A} d_{\mu\nu} \bary_{\mu} / \sum_{\mu \in A}
\bary_{\mu}$ is the average distance from $\nu$ to a facility in $A$. 


\begin{lemma}\label{lem: expected distance in EBGS}
  For any non-empty set $A\subseteq \wbarN(\nu)$, let $\Lambda^\nu_A$ be
  the event that at least one facility in $A$ is opened by Algorithm
  {\EBGS}, and denote by $C_\nu(A)$ the random variable representing
  the distance from $\nu$ to the closest open facility in $A$.  Then
  the expected distance from $\nu$ to the nearest open facility in
  $A$, conditioned on at least one facility in $A$ being opened, is
\begin{equation*}
	\Exp[C_\nu(A) \mid \Lambda^\nu_A ] \le D(A,\nu).
\end{equation*}
\end{lemma}

\begin{proof}
  The proof follows the same reasoning as the proof of Lemma~\ref{lem:
    echs expected C_nu}, so we only sketch it here. We start with a
  similar grouping of facilities in $A$: for each primary demand
  $\kappa$, if $\wbarclsnb(\kappa)\cap A\neq\emptyset$ then
  $\wbarclsnb(\kappa)\cap A$ forms a group. Facilities in $A$ that are
  not in a neighborhood of any primary demand form singleton groups.
  We denote these groups $G_1,...,G_k$. It is clear that the groups
  are disjoint because of (PD'.\ref{PD1:disjoint}). Denoting by
  $\bard_s = D(G_s, \nu)$ the average distance from $\nu$ to a group $G_s$, we
  can assume that these groups are ordered so that $\bard_1\le ... \le
  \bard_k$.

  Each group can have at most one facility open and the events
  representing opening of any two facilities that belong to different
  groups are independent. To estimate the distance from $\nu$ to the
  nearest open facility in $A$, we use an alternative
  random process to make connections, that is easier to
  analyze. Instead of connecting $\nu$ to the nearest open facility in
  $A$, we will choose the smallest $s$ for which $G_s$ has an open
  facility and connect $\nu$ to this facility. (Thus we selected an
  open facility with respect to the minimum $\bard_s$, not the actual
  distance from $\nu$ to this facility.)  This can only increase the
  expected connection cost, thus denoting $g_s = \sum_{\mu\in G_s}
  \gamma\bary_\mu$ for all $s=1,\ldots,k$, and letting $\Prob[\Lambda^\nu_A]$
  be the probability that $A$ has at least one facility open, we have
\begin{align}
    \Exp[C_\nu(A) \mid \Lambda^\nu_A] &\leq \frac{1}{\Prob[\Lambda^\nu_A]} (\bard_1 g_1 +
    \bard_2 g_2 (1 - g_1) + \ldots + \bard_k  g_k(1 -
    g_1)\ldots(1-g_{k-1}))
    \label{eqn: dist set to nu 1}
    \\
    &\leq \frac{1}{\Prob[\Lambda^\nu_A]} \frac{\sum_{s=1}^k \bard_s
      g_s}{\sum_{s=1}^k  g_s} (1 - \prod_{s=1}^k (1 -  g_s))
    \label{eqn: dist set to nu 2}
    \\
    \notag
    &= \frac{\sum_{s=1}^k \bard_s g_s}{\sum_{s=1}^k g_s} =
    \frac{\sum_{\mu \in A} d_{\mu\nu} \gamma \bary_{\mu}}{\sum_{\mu
        \in A} \gamma \bary_{\mu}}
    \\
    \notag
    &= \frac{\sum_{s=1}^k d_{\mu\nu} \bary_{\mu}}{\sum_{\mu \in A}
      \bary_{\mu}} = D(A, \nu).
    \\
    \notag
\end{align}
Inequality (\ref{eqn: dist set to nu 2}) follows from inequality
(\ref{eq:min expected distance}) in~\ref{sec: ECHSinequality}. The rest of the
derivation follows from $\Prob[\Lambda^\nu_A] = 1 - \prod_{s=1}^k (1 -
g_s)$, and the definition of $\bard_s$, $g_s$ and $D(A,\nu)$.
\end{proof}

A consequence of Lemma~\ref{lem: expected distance in EBGS} is the
following corollary which bounds the other two expectations
of $C_\nu$, when at least one facility is opened in $\wbarclsnb(\nu)$,
and when no facility in $\wbarclsnb(\nu)$ opens but a facility in
$\wbarfarnb(\nu)$ is opened.


\begin{corollary} \label{coro: EBGS close and far distance} 
{\rm (a)} $\Exp[C_{\nu} \mid \Lambda_{\cls}^\nu] \leq \clsdist(\nu)$,
and
{\rm (b)} $\Exp[C_{\nu} \mid \Lambda^\nu \wedge \neg \Lambda_{\cls}^\nu]
    			\leq \fardist(\nu)$.
\end{corollary}

\begin{proof}
When there is an open facility in $\wbarclsnb(\nu)$, the algorithm
  connect $\nu$ to the nearest open facility in
  $\wbarclsnb(\nu)$. When no facility in $\wbarclsnb(\nu)$ opens but
  some facility in $\wbarfarnb(\nu)$ opens, the algorithm connects
  $\nu$ to the nearest open facility in $\wbarfarnb(\nu)$. The rest of
  the proof follows from Lemma~\ref{lem: expected distance in
    EBGS}. By setting the set $A$ in Lemma~\ref{lem: expected distance
    in EBGS} to $\wbarclsnb(\nu)$, we have
  \begin{equation*}
    \Exp[C_{\nu} \mid \Lambda_{\cls}^\nu] \leq D(\wbarclsnb(\nu), \nu),
    = \clsdist(\nu),
    \label{eqn: expected connection cost close facility}
  \end{equation*}
proving part (a), and by setting the set $A$ to $\wbarfarnb(\nu)$, we have
  \begin{equation*}
    \Exp[C_{\nu}
    \mid \Lambda^\nu \wedge \neg \Lambda_{\cls}^\nu] \leq
    D(\wbarfarnb(\nu), \nu) = \fardist(\nu),
    \label{eqn: expected connection cost far facility}
  \end{equation*}
which proves part (b).
\end{proof}

Given the estimate on the three expected distances when $\nu$ connects
to its close facility in $\wbarclsnb(\nu)$ in (\ref{eqn: expected
  connection cost close facility}), or its far facility in
$\wbarfarnb(\nu)$ in (\ref{eqn: expected connection cost far
  facility}), or its target facility $\phi(\kappa)$ in (\ref{eqn:
  expected connection cost target facility}), the only missing pieces
are estimates on the corresponding probabilities of each event, which
we do in the next lemma. Once done, we shall put all pieces together
and proving the desired inequality on $\Exp[C_{\nu}]$, that is
(\ref{eqn: expectation of C_nu for EBGS}).

The next Lemma bounds the probabilities for events
that no facilities in $\wbarclsnb(\nu)$ and $\wbarN(\nu)$ are
opened by the algorithm.


\begin{lemma}\label{lem: close and far neighbor probability}
{\rm (a)} $\Prob[\neg\Lambda^\nu_{\cls}] \le 1/e$, and
{\rm (b)} $\Prob[\neg\Lambda^\nu] \le 1/e^\gamma$.
\end{lemma}

\begin{proof}
  (a) To estimate $\Prob[\neg\Lambda^\nu_{\cls}]$, we again consider a
  grouping of facilities in $\wbarclsnb(\nu)$, as in the proof of
  Lemma~\ref{lem: expected distance in EBGS}, according to the primary
  demand's close neighborhood that they fall in, with facilities not
  belonging to such neighborhoods forming their own singleton groups.
  As before, the groups are denoted $G_1, \ldots, G_k$. It is easy to
  see that $\sum_{s=1}^k g_s = \sum_{\mu \in \wbarclsnb(\nu)} \gamma
  \bary_{\mu} = 1$. For any group $G_s$, the probability that a
  facility in this group opens is $\sum_{\mu \in G_s} \gamma
  \bary_{\mu} = g_s$ because in the algorithm at most one facility in
  a group can be chosen and each is chosen with probability $\gamma
  \bary_{\mu}$. Therefore the probability that no facility 
  opens is $\prod_{s=1}^k (1 - g_s)$, which is
  at most $e^{-\sum_{s=1}^k g_s} = 1/e$. Therefore we have
  $\Prob[\neg\Lambda^\nu_A] \leq 1/e$.

(b)
  This proof is similar to the proof of (a). The probability $\Prob[\neg\Lambda^\nu]$ is at most
  $e^{-\sum_{s=1}^k g_s} = 1/e^\gamma$, because we now have
  $\sum_{s=1}^k g_s = \gamma \sum_{\mu \in \wbarN(\nu)} \bary_{\mu} =
  \gamma \cdot 1 = \gamma$.
\end{proof}

We are now ready to bound the overall connection cost of
Algorithm~{\EBGS}, namely inequality (\ref{eqn: expectation of C_nu for EBGS}).


\begin{lemma}\label{lem: EBGS nu's connection cost}
The expected connection of $\nu$ is
\begin{equation*}
\Exp[C_\nu] \le
  \concost(\nu)\cdot\max\Big\{\frac{1/e+1/e^\gamma}{1-1/\gamma}, 1+\frac{2}{e^\gamma}\Big\}.
\end{equation*}
\end{lemma}

\begin{proof}
  Recall that, to connect $\nu$, the algorithm uses the closest facility in
  $\wbarclsnb(\nu)$ if one is opened; otherwise it will try to connect $\nu$
  to the closest facility in $\wbarfarnb(\nu)$. Failing that, it will
  connect $\nu$ to $\phi(\kappa)$, the sole facility open in the
  neighborhood of $\kappa$, the primary demand $\nu$ was assigned
  to. Given that, we estimate $\Exp[C_\nu]$ as follows:
  \begin{align}
    \Exp[C_{\nu}] 
		\;&= \;\Exp[C_{\nu}\mid \Lambda^\nu_{\cls}] \cdot \Prob[\Lambda^\nu_{\cls}]	
				\;+\; \Exp[C_{\nu}\mid \Lambda^\nu\ \wedge\neg \Lambda^\nu_{\cls}] 
				\cdot \Prob[\Lambda^\nu\, \wedge\neg \Lambda^\nu_{\cls}]	
				\notag
		\\
		& \quad\quad\quad
				+ \; \Exp[C_{\nu}\mid \neg \Lambda^\nu] \cdot \Prob[\neg \Lambda^\nu]
				\notag
		\\
		&\leq \; \clsdist(\nu) \cdot \Prob[\Lambda^\nu_{\cls}]
			\;+\; \fardist(\nu)	
				\cdot \Prob[\Lambda^\nu\, \wedge\neg \Lambda^\nu_{\cls}]
                      \label{eqn: apply three expected dist}
						\\
                        &\quad\quad\quad
			+\; [\,\clsdist(\nu) + 2\fardist(\nu)\,] \cdot \Prob[\neg\Lambda^\nu]
		\notag
		\\
                &=\; [\,\clsdist(\nu) + \fardist(\nu)\,]\cdot \Prob[\neg\Lambda^\nu] 
						\;+\; 
							[\,\fardist(\nu)   -\clsdist(\nu)\,]
                                \cdot \Prob[\neg\Lambda^\nu_{\cls}]
                              \;+\;  \clsdist(\nu)
                                                        \notag
		\\
             &\leq\; [\,\clsdist(\nu) + \fardist(\nu)\,] \cdot \frac{1}{e^\gamma}
             \;+\; [\,\fardist(\nu) - \clsdist(\nu)\,] \cdot \frac{1}{e}
             \;+\; \clsdist(\nu)
             \label{eqn: probability estimate}
             \\
             \notag
             &=\; \Big(1 - \frac{1}{e} + \frac{1}{e^\gamma}\Big)\cdot \clsdist(\nu)
 				\;+\; \Big(\frac{1}{e} + \frac{1}{e^\gamma}\Big)\cdot\fardist(\nu).
\end{align}
Inequality (\ref{eqn: apply three expected dist}) follows from
Corollary~\ref{coro: EBGS close and far distance} and 
Lemma~\ref{lem: EBGS target connection cost}. 
Inequality (\ref{eqn: probability estimate}) follows from 
Lemma~\ref{lem: close and far neighbor probability} and
$\fardist(\nu) - \clsdist(\nu)\ge 0$.

Now define $\rho =\clsdist(\nu)/\concost(\nu)$. It is easy to
see that $\rho$ is between 0 and 1. Continuing the above
derivation, applying (\ref{eqn:avg dist cls dist far dist}), we get
\begin{align*}
\Exp[C_{\nu}]
             \;&\le\; \concost(\nu) 
			\cdot\left((1-\rho)\frac{1/e+1/e^\gamma}{1-1/\gamma} 
				+ \rho (1 + \frac{2}{e^\gamma})\right)
			\\
             &\leq \concost(\nu) 
				\cdot \max\left\{\frac{1/e+1/e^\gamma}{1-1/\gamma}, 1 + \frac{2}{e^\gamma}\right\},
\end{align*}
and the proof is now complete.
\end{proof}

With Lemma~\ref{lem: EBGS nu's connection cost} proven, we are now ready to bound our total connection cost.
For any client $j$ we have
\begin{align*}
\sum_{\nu\in j} C^{\avg}(\nu)
	&= \sum_{\nu\in j}\sum_{\mu\in\facilityset} d_{\mu\nu}\barx_{\mu\nu} 
	\\
	&= \sum_{i\in\sitesset}d_{ij}\sum_{\mu\in i}\sum_{\nu\in j} \barx_{\mu\nu}
	= \sum_{i\in\sitesset} d_{ij}x_{ij}^\ast = C_j^\ast.
\end{align*}
Summing over all clients $j$ we obtain that the total expected connection cost is
\begin{equation*}
	\Exp[ C_{\smallEBGS} ] \le  C^\ast\max\left\{\frac{1/e+1/e^\gamma}{1-1/\gamma}, 1+\frac{2}{e^\gamma}\right\}.
\end{equation*}
Recall that the expected facility cost is bounded by $\gamma F^\ast$,
as argued earlier. Hence the total expected cost is bounded by $\max\{\gamma,
\frac{1/e+1/e^\gamma}{1-1/\gamma}, 1+\frac{2}{e^\gamma}\}\cdot
\LP^\ast$. Picking $\gamma=1.575$ we obtain the desired ratio.


\begin{theorem}\label{thm:ebgs}
  Algorithm~{\EBGS} is a $1.575$-approximation algorithm for \FTFP.
\end{theorem}



\section{Final Comments}

In this paper we show a sequence of LP-rounding approximation algorithms
for FTFP, with the best algorithm achieving  ratio $1.575$. 
As we mentioned earlier, we believe that 
our techniques of demand reduction and adaptive partitioning are very flexible and
should be useful in extending other LP-rounding methods for UFL to obtain
matching bounds for FTFP.

One of the main open problems in this area is whether FTFL can be approximated with the
same ratio as UFL, and our work was partly motivated by this question. The techniques we
introduced are not directly applicable to FTFL, mainly because our partitioning
approach involves facility splitting that could result in several sibling demands being served
by facilities on the same site. Nonetheless, we hope that further refinements of 
our construction might get around this issue and
lead to new algorithms for FTFL with improved ratios.
\pagebreak

\bibliographystyle{elsarticle-num}
\bibliography{facility}

\begin{thebibliography}{10}

\bibitem{AryaGKMMP04}
Vijay Arya, Naveen Garg, Rohit Khandekar, Adam Meyerson, Kamesh Munagala, and
  Vinayaka Pandit.
\newblock Local search heuristics for k-median and facility location problems.
\newblock {\em SIAM J. Comput.}, 33(3):544--562, 2004.

\bibitem{ByrkaA10}
Jaroslaw Byrka and Karen Aardal.
\newblock An optimal bifactor approximation algorithm for the metric
  uncapacitated facility location problem.
\newblock {\em SIAM J. Comput.}, 39(6):2212--2231, 2010.

\bibitem{ByrkaGS10}
Jaroslaw Byrka, MohammadReza Ghodsi, and Aravind Srinivasan.
\newblock {LP}-rounding algorithms for facility-location problems.
\newblock {\em CoRR}, abs/1007.3611, 2010.

\bibitem{ByrkaSS10}
Jaroslaw Byrka, Aravind Srinivasan, and Chaitanya Swamy.
\newblock Fault-tolerant facility location, a randomized dependent
  {LP}-rounding algorithm.
\newblock In {\em Proceedings of the 14th Integer Programming and Combinatorial
  Optimization, IPCO '10}, pages 244--257, 2010.

\bibitem{ChudakS04}
Fabi\'{a}n Chudak and David Shmoys.
\newblock Improved approximation algorithms for the uncapacitated facility
  location problem.
\newblock {\em SIAM J. Comput.}, 33(1):1--25, 2004.

\bibitem{GuhaK98}
Sudipto Guha and Samir Khuller.
\newblock Greedy strikes back: improved facility location algorithms.
\newblock In {\em Proceedings of the 9th ACM-SIAM Symposium on Discrete
  Algorithms, SODA '98}, pages 649--657, 1998.

\bibitem{GuhaMM01}
Sudipto Guha, Adam Meyerson, and Kamesh Munagala.
\newblock Improved algorithms for fault tolerant facility location.
\newblock In {\em Proceedings of the 12th ACM-SIAM Symposium on Discrete
  Algorithms, SODA '01}, pages 636--641, 2001.

\bibitem{gupta08}
Anupam Gupta.
\newblock Lecture notes: {{CMU}} 15-854b, spring 2008, 2008.

\bibitem{HardyLP88}
Godfrey~Harold Hardy, John~Edensor Littlewood, and George P\'{o}lya.
\newblock {\em Inequalities}.
\newblock Cambridge University Press, 1988.

\bibitem{JainMMSV03}
Kamal Jain, Mohammad Mahdian, Evangelos Markakis, Amin Saberi, and Vijay
  Vazirani.
\newblock Greedy facility location algorithms analyzed using dual fitting with
  factor-revealing {LP}.
\newblock {\em J. ACM}, 50(6):795--824, 2003.

\bibitem{JainV01}
Kamal Jain and Vijay Vazirani.
\newblock Approximation algorithms for metric facility location and k-median
  problems using the primal-dual schema and lagrangian relaxation.
\newblock {\em J. ACM}, 48(2):274--296, 2001.

\bibitem{JainV03}
Kamal Jain and Vijay~V. Vazirani.
\newblock An approximation algorithm for the fault tolerant metric facility
  location problem.
\newblock {\em Algorithmica}, 38(3):433--439, 2003.

\bibitem{Li11}
Shi Li.
\newblock A 1.488 approximation algorithm for the uncapacitated facility
  location problem.
\newblock In {\em Proceedings of the 38th International Conference on Automata,
  Languages and Programming, ICALP '11}, volume 6756, pages 77--88, 2011.

\bibitem{LiaoShen11}
Kewen Liao and Hong Shen.
\newblock Unconstrained and constrained fault-tolerant resource allocation.
\newblock In {\em Proceedings of the 17th Annual International Conference on
  Computing and Combinatorics}, COCOON'11, pages 555--566, 2011.

\bibitem{MahdianYZ06}
Mohammad Mahdian, Yinyu Ye, and Jiawei Zhang.
\newblock Approximation algorithms for metric facility location problems.
\newblock {\em SIAM J. Comput.}, 36(2):411--432, 2006.

\bibitem{ShmoysTA97}
David Shmoys, \'{E}va Tardos, and Karen Aardal.
\newblock Approximation algorithms for facility location problems (extended
  abstract).
\newblock In {\em Proceedings of the 29th Annual ACM Symposium on Theory of
  Computing, STOC '97}, pages 265--274, 1997.

\bibitem{Svi02}
Maxim Sviridenko.
\newblock An improved approximation algorithm for the metric uncapacitated
  facility location problem.
\newblock In {\em Proceedings of the 9th International IPCO Conference on
  Integer Programming and Combinatorial Optimization, IPCO '02}, pages
  240--257, 2002.

\bibitem{SwamyS08}
Chaitanya Swamy and David Shmoys.
\newblock Fault-tolerant facility location.
\newblock {\em ACM Trans. Algorithms}, 4(4):1--27, 2008.

\bibitem{vygen05}
Jens Vygen.
\newblock {\em Approximation Algorithms for Facility Location Problems}.
\newblock Forschungsinst. f\"{u}r Diskrete Mathematik, 2005.

\bibitem{XuS09}
Shihong Xu and Hong Shen.
\newblock The fault-tolerant facility allocation problem.
\newblock In {\em Proceedings of the 20th International Symposium on Algorithms
  and Computation, ISAAC '09}, pages 689--698, 2009.

\bibitem{YanC11}
Li~Yan and Marek Chrobak.
\newblock Approximation algorithms for the fault-tolerant facility placement
  problem.
\newblock {\em Inf. Process. Lett.}, 111(11):545--549, 2011.

\end{thebibliography}

\pagebreak

\appendix


\section{Proof of Lemma~\ref{lem: EBGS target connection cost}}\label{sec: proof of lemma 15}

Lemma~\ref{lem: EBGS target connection cost} provides a bound on the
expected connection cost of a demand $\nu$ when Algorithm~{\EBGS} does not open
any facilities in $\wbarN(\nu)$, namely
\begin{equation}
  \Exp[C_{\nu} \mid \neg \Lambda^{\nu}] \leq
  \clsdist(\nu) +  2\fardist(\nu),
			\label{eqn: lemma ebgs target connection cost}
\end{equation}
We show a stronger inequality that 
\begin{equation}
  \Exp[C_{\nu} \mid \neg \Lambda^{\nu}] \leq
  \clsdist(\nu) + \clsmax(\nu) + \fardist(\nu)
			\label{eqn: lemma ebgs indirect connection cost},
\end{equation}
which then implies (\ref{eqn: lemma ebgs target connection cost})
because $\clsmax(\nu) \leq \fardist(\nu)$.  The proof of (\ref{eqn:
  lemma ebgs indirect connection cost}) is similar to that in
\cite{ByrkaA10}. For the sake of completeness, we provide it here,
formulated in our terminology and notation.

Assume that the event $\neg \Lambda^{\nu}$ is true, that is Algorithm~{\EBGS}
does not open any facility in $\wbarN(\nu)$.
Let $\kappa$ be the primary demand that $\nu$ was assigned to. Also let
\begin{equation*}
K = \wbarclsnb(\kappa) \setminus \wbarN(\nu), \quad
V_{\cls} = \wbarclsnb(\kappa) \cap \wbarclsnb(\nu) \quad \textrm{and}\quad 
V_{\far} = \wbarclsnb(\kappa) \cap \wbarfarnb(\nu).
\end{equation*}
Then $K, V_{\cls}, V_{\far}$ form a partition of
$\wbarclsnb(\kappa)$, that is, they are disjoint and their union is $\wbarclsnb(\kappa)$.
Moreover, we have that $K$ is not empty, because Algorithm~{\EBGS}
opens some facility in $\wbarclsnb(\kappa)$ and this facility cannot be in $V_{\cls}\cup V_{\far}$,
by our assumption. 
We also have that $V_{\cls}$ is not empty due to (PD'.\ref{PD1:assign:overlap}). 

Recall that $D(A,\eta) = \sum_{\mu\in A}d_{\mu\eta}\bary_{\mu}/\sum_{\mu\in A}\bary_{\mu}$
is the average distance between a demand $\eta$ and the facilities in a set $A$. We shall show that
\begin{equation}
	 D(K, \nu) \leq \clsdist(\kappa)+\clsmax(\kappa) + \fardist(\nu).
				\label{eqn: bound on D(K,nu)}
\end{equation}
This is sufficient, because, by the algorithm, $D(K,\nu)$ is exactly 
the expected connection cost for demand $\nu$ conditioned on
the event that none of $\nu$'s neighbors 
opens, that is the left-hand side of (\ref{eqn: lemma ebgs indirect connection cost}).
Further, (PD'.\ref{PD1:assign:cost}) states that 
$\clsdist(\kappa)+\clsmax(\kappa) \le \clsdist(\nu) + \clsmax(\nu)$, and thus
(\ref{eqn: bound on D(K,nu)})  implies (\ref{eqn: lemma ebgs indirect connection cost}).

\medskip

The proof of (\ref{eqn: bound on D(K,nu)}) is by analysis of several cases.

\medskip
\noindent
{\mycase{1}} $D(K, \kappa) \leq \clsdist(\kappa)$. For any
facility $\mu \in V_{\cls}$ (recall that $V_{\cls}\neq\emptyset$), 
we have $d_{\mu\kappa} \leq \clsmax(\kappa)$ 
and $d_{\mu\nu} \leq \clsmax(\nu) \leq \fardist(\nu)$. Therefore, using the
case assumption, we get
	$D(K,\nu) \leq D(K,\kappa) + d_{\mu\kappa} + d_{\mu\nu} 
				\leq \clsdist(\kappa) + \clsmax(\kappa) + \fardist(\nu)$.

\medskip
\noindent
{\mycase{2}} There exists a facility $\mu\in V_{\cls}$ such that
  $d_{\mu\kappa} \leq \clsdist(\kappa)$. Since $\mu\in V_{\cls}$, we infer
  that $d_{\mu\nu} \leq \clsmax(\nu) \leq \fardist(\nu)$.  Using
  $\clsmax(\kappa)$ to bound $D(K, \kappa)$, we have $D(K, \nu)
  \leq D(K, \kappa) + d_{\mu\kappa} + d_{\mu\nu} \leq
  \clsmax(\kappa) + \clsdist(\kappa) + \fardist(\nu)$.

\medskip
\noindent
{\mycase{3}} In this case we assume that neither of Cases~1 and 2 applies, that is
 $D(K, \kappa) > \clsdist(\kappa)$ and every $\mu \in V_{\cls}$ satisfies
 $d_{\mu\kappa} >  \clsdist(\kappa)$. This implies that
$D(K\cup V_{\cls}, \kappa) > \clsdist(\kappa) = D(\wbarclsnb(\kappa), \kappa)$.
Since sets $K$, $V_{\cls}$ and $V_{\far}$ form a partition of $\wbarclsnb(\kappa)$,
we obtain that in this case $V_{\far}$ is not
empty and $D(V_{\far}, \kappa) < \clsdist(\kappa)$. 
Let $\delta = \clsdist(\kappa) - D(V_{\far}, \kappa) > 0$. 
We now have two sub-cases:
\begin{description}
	
\item{\mycase{3.1}} {$D(V_{\far}, \nu) \leq \fardist(\nu) + \delta$}.
  Substituting $\delta$, this implies that $D(V_{\far}, \nu) +
  D(V_{\far},\kappa) \le \clsdist(\kappa) + \fardist(\nu)$.  From the
  definition of the average distance $D(V_{\far},\kappa)$ and
  $D(V_{\far}, \nu)$, we obtain that there exists some $\mu \in
  V_{\far}$ such that $d_{\mu\kappa} + d_{\mu\nu} \leq
  \clsdist(\kappa) + \fardist(\nu)$.  Thus $D(K, \nu) \leq D(K,
  \kappa) + d_{\mu\kappa} + d_{\mu\nu} \leq \clsmax(\kappa) +
  \clsdist(\kappa) + \fardist(\nu)$.

\item{\mycase{3.2}} {$D(V_{\far}, \nu) > \fardist(\nu) + \delta$}.
  The case assumption implies that $V_{\far}$ is a proper subset of
  $\wbarfarnb(\nu)$, that is $\wbarfarnb(\nu) \setminus V_{\far}
  \neq\emptyset$.  Let $\hat{y} = \gamma \sum_{\mu\in V_{\smallfar}}
  \bary_{\mu}$.  We can express $\fardist(\nu)$ using $\hat{y}$ as
  follows
\begin{equation*}
\fardist(\nu) = D(V_{\far},\nu) \frac{\hat{y}}{\gamma-1} +
    D(\wbarfarnb(\nu)\setminus V_{\far}, \nu) \frac{\gamma-1-\hat{y}}{\gamma-1}.
\end{equation*}
Then, using the case condition and simple algebra, we have
  \begin{align}
    \clsmax(\nu) &\leq D(\wbarfarnb(\nu) \setminus V_{\far}, \nu) 
			\notag
		\\
		&\leq \fardist(\nu) - \frac{\hat{y}\delta}{\gamma-1-\hat{y}} 
		\leq \fardist(\nu) - \frac{\hat{y}\delta}{1-\hat{y}},
			\label{eqn: case 3, bound on C_cls^max(nu)}
  \end{align}
where the last step follows from $1 < \gamma < 2$. 

On the other hand, since $K$, $V_{\cls}$, and $V_{\far}$ form a partition of $\wbarclsnb(\kappa)$,
we have
$\clsdist(\kappa) = (1-\hat{y}) D(K\cup V_{\cls}, \kappa) + \hat{y} D(V_{\far}, \kappa)$.
Then using the definition of $\delta$ we obtain
\begin{equation}
    D(K \cup V_{\cls}, \kappa) = \clsdist(\kappa) + \frac{\hat{y}\delta}{1-\hat{y}}.
				\label{eqn: formula for D(V_cls,kappa)}
\end{equation}
  Now we are essentially done. If there exists some $\mu \in V_{\cls}$ such
  that $d_{\mu\kappa} \leq \clsdist(\kappa) +
  \hat{y}\delta/(1-\hat{y})$, then	we have
  \begin{align*}
    D(K, \nu) &\leq D(K, \kappa) + d_{\mu\kappa} + d_{\mu\nu} \\
    &\leq \clsmax(\kappa) + \clsdist(\kappa) +
    			\frac{\hat{y}\delta}{1-\hat{y}}
    + \clsmax(\nu)\\
    &\leq \clsmax(\kappa) + \clsdist(\kappa) + \fardist(\nu),
  \end{align*}
where we used (\ref{eqn: case 3, bound on C_cls^max(nu)}) in the last step.
  Otherwise, from (\ref{eqn: formula for D(V_cls,kappa)}),
we must have $D(K, \kappa) \leq \clsdist(\kappa) +
  \hat{y}\delta/(1-\hat{y})$. Choosing any $\mu \in V_{\cls}$, it follows that
  \begin{align*}
    D(K, \nu) &\leq D(K, \kappa) + d_{\mu\kappa} + d_{\mu\nu} \\
    &\leq \clsdist(\kappa) + \frac{\hat{y}\delta}{1-\hat{y}} +
    		\clsmax(\kappa)  + \clsmax(\nu)\\
    &\leq \clsdist(\kappa) + \clsmax(\kappa) + \fardist(\nu),
  \end{align*}
again using (\ref{eqn: case 3, bound on C_cls^max(nu)}) in the last step.

\end{description}

This concludes the proof of (\ref{eqn: lemma ebgs target connection cost}).
As explained earlier, Lemma~\ref{lem: EBGS target connection cost} follows.


\vfill

\section{Proof of Inequality  (\ref{eqn: echs ineq direct cost, step 1})}
\label{sec: ECHSinequality}

In Sections~\ref{sec: 1.736-approximation} and \ref{sec: 1.575-approximation}
we use the following inequality
\begin{align}
  \label{eq:min expected distance}
  \bard_1 g_1 + \bard_2 g_2 (1-g_1) +
  \ldots &+ \bard_k g_k (1-g_1) (1-g_2) \ldots (1-g_k)\\ \notag
  &\leq \frac{1}{\sum_{s=1}^k g_s} \left(\textstyle\sum_{s=1}^k \bard_s g_s\right)\left(\textstyle\sum_{t=1}^k g_t \textstyle\prod_{z=1}^{t-1} (1-g_z)\right).
\end{align}
for $0 < \bard_1\leq \bard_2 \leq \ldots \leq \bard_k$, and
$0 < g_1,...,g_s \le 1$.

\medskip

We give here a new proof of this inequality, much simpler
than the existing proof in \cite{ChudakS04}, and also simpler than the
argument by Sviridenko~\cite{Svi02}.  We derive this inequality from
the following generalized version of the Chebyshev Sum Inequality:
\begin{equation}
  \label{eq:cheby}
  \textstyle{\sum_{i}} p_i \textstyle{\sum_j} p_j a_j b_j \leq \textstyle{\sum_i} p_i a_i \textstyle{\sum_j} p_j b_j,
\end{equation}
where each summation runs from $1$ to $l$ and the sequences $(a_i)$,
$(b_i)$ and $(p_i)$ satisfy the following conditions: $p_i\geq 0, a_i
\geq 0, b_i \geq 0$ for all $i$, $a_1\leq a_2 \leq \ldots \leq a_l$,
and $b_1 \geq b_2 \geq \ldots \geq b_l$.

Given inequality (\ref{eq:cheby}), we can obtain our inequality
(\ref{eq:min expected distance}) by simple substitution
\begin{equation*}
  p_i \leftarrow g_i, a_i \leftarrow \bard_i, b_i \leftarrow
  \Pi_{s=1}^{i-1} (1-g_s),
\end{equation*}
for $i = 1,...,k$.

\ignore{
For the sake of completeness, we include the proof of inequality (\ref{eq:cheby}), 
due to Hardy, Littlewood and Polya~\cite{HardyLP88}. The idea is to evaluate the 
following sum:
\begin{align*}
  S &= \textstyle{\sum_i} p_i \textstyle{\sum_j} p_j a_j b_j - \textstyle{\sum_i} p_i a_i \textstyle{\sum_j} p_j b_j
	\\
  & = \textstyle{\sum_i \sum_j} p_i p_j a_j b_j - \textstyle{\sum_i \sum_j} p_i a_i p_j b_j
	\\
  & = \textstyle{\sum_j \sum_i} p_j p_i a_i b_i - \textstyle{\sum_j \sum _i} p_j a_j p_i b_i
	\\
	&= \half \cdot \textstyle{\sum_i \sum_j} (p_i p_j a_j b_j - p_i a_i p_j b_j + p_j p_i a_i
  							b_i - p_j a_j p_i b_i)
\\
  &= \half \cdot \textstyle{\sum_i \sum_j} p_i p_j (a_i - a_j)(b_i - b_j) \leq 0.
\end{align*}
The last inequality holds because $(a_i-a_j)(b_i-b_j) \leq 0$, since the sequences
$(a_i)$ and $(b_i)$ are ordered oppositely.
}

\end{document}